\documentclass[11pt]{article}
\usepackage[a4paper,margin=2.54cm]{geometry}

\usepackage{latexsym}
\usepackage{xspace}
\usepackage{amssymb}
\usepackage{stmaryrd}
\usepackage{hyperref}
\usepackage{url}
\usepackage{graphicx}
\input{xy}
\xyoption{all}
\usepackage{calc}
\usepackage{tikz}
\usepackage{flushend}
\usepackage{color}
\usepackage{cite}
\usepackage{authblk}

\newcommand{\email}[1]{\mbox{Email: \url{#1}}}

\newtheorem{theorem}{Theorem}[section]
\newtheorem{lemma}[theorem]{Lemma}

\newtheorem{corollary}[theorem]{Corollary}

\newtheorem{Definition}[theorem]{Definition}
\newtheorem{Example}[theorem]{Example}
\newtheorem{Remark}[theorem]{Remark}

\newenvironment{definition}{\begin{Definition}\begin{em}}{\end{em}\end{Definition}}
\newenvironment{example}{\begin{Example}\begin{em}}{\end{em}\end{Example}}
\newenvironment{remark}{\begin{Remark}\begin{em}}{\end{em}\end{Remark}}
\newenvironment{proof}{
	
	\smallskip
	
	\noindent
	{\em Proof.}}{
	
	\smallskip
	
}

\def\eqref#1{(\ref{#1})}

\def\tuple#1{\langle#1\rangle}

\newcommand{\mL}{\mathcal{L}}
\newcommand{\bL}{\mathbf{L}}

\newcommand{\mA}{\mathcal{A}}
\newcommand{\mAp}{{\mathcal{A}'\!}}
\newcommand{\mAdp}{{\mathcal{A}''\!}}

\newcommand{\myend}{\mbox{}\hfill{\small$\blacksquare$}}

\newcommand{\comment}[1]{}

\newcommand{\fand}{\varotimes}

\newcommand{\fto}{\Rightarrow}
\newcommand{\fequiv}{\Leftrightarrow}

\newcommand{\deltaA}{\delta^\mA}
\newcommand{\sigmaA}{\sigma^\mA}
\newcommand{\tauA}{\tau^\mA}
\newcommand{\deltaAp}{\delta^\mAp}
\newcommand{\sigmaAp}{\sigma^\mAp}
\newcommand{\tauAp}{\tau^\mAp}
\newcommand{\deltaAdp}{\delta^\mAdp}
\newcommand{\sigmaAdp}{\sigma^\mAdp}
\newcommand{\tauAdp}{\tau^\mAdp}

\newcommand{\normS}[3]{\|#1\|_{{#2} \fto {#3}}}
\newcommand{\normBS}[3]{\|#1\|_{{#2} \fequiv {#3}}}

\newcommand{\nZs}{\normS{\varphi}{\mA}{\mAp}}
\newcommand{\nZbs}{\normBS{\varphi}{\mA}{\mAp}}

\newcommand{\mF}{\mathcal{F}}
\newcommand{\mFs}{\mF_{\!\!_\to}(\Sigma,\mL)}
\newcommand{\mFbs}{\mF_{\!\!_\leftrightarrow}(\Sigma,\mL)}

\begin{document}
\sloppy
	
\title{Fuzzy Simulations and Bisimulations between Fuzzy Automata}

\author{Linh Anh Nguyen}

\affil{\small Institute of Informatics, University of Warsaw, Banacha 2, 02-097 Warsaw, Poland, \email{nguyen@mimuw.edu.pl}}

\affil{\small
	Faculty of Information Technology, Nguyen Tat Thanh University, Ho Chi Minh City, Vietnam
}

\date{}

\maketitle

\begin{abstract}
Simulations and bisimulations between two fuzzy automata over a complete residuated lattice were defined by {\'C}iri{\'c} et al.~(2012) as fuzzy relations between the sets of states of the automata. However, they act as a crisp relationship between the automata. In particular, if there exists a (forward) bisimulation between two fuzzy automata, then the fuzzy languages recognized by them are crisply equal. 
Approximate simulations and bisimulations introduced by Stanimirovi{\' c} et al.~(2020) aim at fuzzifying this phenomenon. However, they are defined only for fuzzy automata over a complete Heyting algebra and do not give the exact relationship between states of the automata. 
In this article, we introduce and study {\em fuzzy} simulations and bisimulations between fuzzy automata over a complete residuated lattice. These notions are novel and have good properties. They are defined for fuzzy automata over any complete residuated lattice. 
We prove that the fuzzy language recognized by a fuzzy automaton is fuzzily preserved by fuzzy simulations and fuzzily invariant under fuzzy bisimulations. We also prove that the notions of fuzzy simulation and bisimulation have the Hennessy-Milner properties, which are a logical characterization of the greatest fuzzy simulation or bisimulation between two fuzzy automata. 
In addition, we provide results showing that our notions of fuzzy simulation and bisimulation are more general and refined than the notions of simulation and bisimulation introduced by {\'C}iri{\'c} et al.\ and the notions of approximate simulation and bisimulation introduced by Stanimirovi{\' c} et al.

\medskip

\noindent {\em Keywords:} fuzzy simulation, fuzzy bisimulation, fuzzy automata, residuated lattice.
\end{abstract}

%===============================================================================

\section{Introduction}
\label{section:intro}

\comment{
TO DO:
\begin{itemize}
\item Rephrase II: B, C, D, E
\item \ldots
\end{itemize}
}

Simulation and bisimulation are well-known notions in computer science~\cite{vBenthem76,Park81,HennessyM85,Sangiorgi09}. They are used, among others, to compare the behaviors of labeled transition systems (LTSs) and specify the logical similarity or indiscernibility between states in Kripke models (see, e.g.~\cite{HennessyM85,BRV2001}). The largest auto-bisimulation of an LTS or a Kripke model is an equivalence relation. It can be exploited to minimize the considered system. The largest auto-bisimulation of an interpretation in a description logic can also be used for concept learning~\cite{LbRoughification}. 

Automata differ from LTSs in that they have initial and terminal states. They are generalized to fuzzy automata by allowing transitions and the sets of initial or terminal states to be fuzzy. 
In~\cite{CiricIDB12} {\'C}iri{\'c} et al.\ introduced two kinds of simulations (forward and backward) and four kinds of bisimulations (forward, backward, forward-backward and backward-forward) between fuzzy automata over a complete residuated lattice. They form pairs of dual notions. Among those kinds, forward simulations and bisimulations can be treated as the default. The work~\cite{CiricIDB12} concentrates on studying forward simulations and bisimulations that are uniform fuzzy relations and provides results on isomorphisms between factor fuzzy automata. 

Simulations and bisimulations introduced in~\cite{CiricIDB12} for fuzzy automata are fuzzy relations between the sets of states of two fuzzy automata. However, they act as a {\em crisp} relationship between the automata. In particular, if there exists a forward simulation (respectively, bisimulation) between fuzzy automata $\mA$ and $\mAp$, then the fuzzy language recognized by $\mA$ is crisply less than or equal (respectively, crisply equal) to the fuzzy language recognized by $\mAp$. 
In~\cite{SMC.20} Stanimirovi{\' c} et al.\ introduced approximate simulations and bisimulations between fuzzy automata with the aim to fuzzify this phenomenon. These notions are also fuzzy relations. However, they are defined only for fuzzy automata over a complete Heyting algebra and do not give the exact relationship between states of the automata. 

The motivation of this work is to introduce novel notions of simulation and bisimulation between fuzzy automata that are more refined than the ones defined in~\cite{CiricIDB12} and~\cite{SMC.20}. 

In this article, we introduce and study {\em fuzzy} simulations and bisimulations between fuzzy automata over a complete residuated lattice. We prove that the fuzzy language recognized by a fuzzy automaton is fuzzily preserved by fuzzy simulations (Theorem~\ref{theorem: HDFUI}) and fuzzily invariant under fuzzy bisimulations (Theorem~\ref{theorem: HDFUI2}). We also prove that the notions of fuzzy simulation and bisimulation have the Hennessy-Milner properties (Theorems~\ref{theorem: HGDJA} and~\ref{theorem: HGDJAt}), which are a logical characterization of the greatest fuzzy simulation or bisimulation between two fuzzy automata. 
Our notions of fuzzy simulation and bisimulation for fuzzy automata are more refined than the notions of forward simulation and bisimulation introduced in~\cite{CiricIDB12} and the notions of approximate forward simulation and bisimulation introduced in~\cite{SMC.20} in the following aspects:
\begin{itemize}
\item Every forward simulation (respectively, bisimulation) between two fuzzy automata is a fuzzy simulation (respectively, bisimulation) between them, but not vice versa. 
\item While there may not exist any forward simulation or bisimulation between two fuzzy automata, there always exist the greatest fuzzy simulation and bisimulation between them. 
\item While the fuzzy language recognized by a fuzzy automaton is crisply preserved by forward simulations and crisply invariant under forward bisimulations \cite[Theorem~5.3]{CiricIDB12}, it is fuzzily preserved by fuzzy simulations and fuzzily invariant under fuzzy bisimulations.

\item While approximate forward simulations and bisimulations~\cite{SMC.20} are defined only for fuzzy automata over a complete Heyting algebra, fuzzy simulations and bisimulations are defined for fuzzy automata over any complete residuated lattice. 

\item Given image-finite fuzzy automata $\mA$ and $\mAp$ over the G\"odel structure and an appropriate threshold~$\lambda$, the greatest $\lambda$-approximate forward simulation (respectively, bisimulation) between $\mA$ and $\mAp$ may not give the exact relationship between states of the automata (see Examples~\ref{example: JHDLS} and~\ref{example: JHHSK}), while the greatest fuzzy simulation (respectively, bisimulation) between $\mA$ and $\mAp$ always gives the exact relationship between states of the automata (as stated by the Hennessy-Milner properties of fuzzy simulations and bisimulations). 
\end{itemize}

The rest of this article is structured as follows. Section~\ref{section: prel} contains preliminaries. 
In Section~\ref{section: fs-4-fa} (respectively, Section~\ref{section: f-s-bs}), we define fuzzy simulations (respectively, bisimulations) between fuzzy automata and present our results on them. Section~\ref{sec: related work} is a discussion on related work. We give concluding remarks in Section~\ref{sec: conc}.

%===============================================================================

\section{Preliminaries}
\label{section: prel}

This section recalls basic definitions about residuated lattices, fuzzy sets and relations, fuzzy automata, forward simulations and bisimulations, as well as approximate forward simulations and bisimulations between fuzzy automata.

\subsection{Residuated Lattices}

A {\em residuated lattice} \cite{Hajek1998,Belohlavek2002} is an algebra $\mL = \tuple{L$, $\leq$, $\fand$, $\fto$, $0$, $1}$ such that%\footnote{See also \url{https://en.wikipedia.org/wiki/Residuated_lattice}} 
\begin{itemize}
\item $\tuple{L, \leq, 0, 1}$ is a lattice with the smallest element 0 and the greatest element 1,
\item $\tuple{L, \fand, 1}$ is a commutative monoid, 
\item for every $a, b, c \in L$,
\begin{equation}
a \fand b \leq c \ \ \textrm{iff}\ \ a \leq (b \fto c). \label{fop: GDJSK 00} 
\end{equation}
\end{itemize}

%The expression \mbox{$b \fto c$} is called the {\em residual} of $c$ by $b$. 
Given a residuated lattice \mbox{$\mL = \tuple{L, \leq, \fand, \fto, 0, 1}$}, let $\lor$ and $\land$ denote the corresponding {\em meet} and {\em join} operators. We write $a \fequiv b$ to denote \mbox{$(a \fto b) \land (b \fto a)$}. For $A \subseteq L$, by $\bigvee\!A$ and $\bigwedge\!A$ we denote the supremum and infimum of $A$, respectively, if they exist. Similarly, for $A = \{a_i \mid i \in I\} \subseteq L$, by $\bigvee_{i \in I} a_i$ and $\bigwedge_{i \in I} a_i$ we denote $\bigvee\!A$ and $\bigwedge\!A$, respectively, if they exist. We assume that $\fand$ and $\land$ bind stronger than $\lor$, which binds stronger than~$\fto$ and~$\fequiv$. 

A residuated lattice $\mL = \tuple{L, \leq, \fand, \fto, 0, 1}$ is {\em complete} if the lattice \mbox{$\tuple{L, \leq, 0, 1}$} is complete. It is a {\em Heyting algebra} if $\fand = \land$. 
We say that $\fand$ is {\em continuous} (with respect to infima) if, for every $a \in L$ and $B \subseteq L$, 
\[ a \fand {\textstyle\bigwedge} B = \bigwedge_{b \in B}\!(a \fand b). \]

%We will need the following lemma. 

\begin{lemma}(cf.\ \cite{Hajek1998,Belohlavek2002,NguyenFSS2021})\label{lemma: JHFJW}
	Let $\mL = \tuple{L, \leq, \fand, \fto, 0, 1}$ be a residuated lattice.
	The following properties hold for all $a,a',b,b',c \in L$:
	\begin{eqnarray}
	\!\!\!\!\!\!\!\!\!\! a \leq a' \textrm{ and } b \leq b' & \!\!\textrm{implies}\!\! & a \fand b \leq a' \fand b' \label{fop: GDJSK 10}\\
	\!\!\!\!\!\!\!\!\!\! a' \leq a \textrm{ and } b \leq b' & \!\!\textrm{implies}\!\! & (a \fto b) \leq (a' \fto b') \label{fop: GDJSK 20}\\
	\!\!\!\!\!\!\!\!\!\! a \leq b & \textrm{iff} & (a \fto b) = 1 \label{fop: GDJSK 30}
	\end{eqnarray}
	\begin{eqnarray}
	\!\!\!\!\!\!\!\!\!\! a \fand 0 & = & 0 \label{fop: GDJSK 40} \\
	%\!\!\!\!\!\!\!\!\!\! a \fand (b \lor c) & = & a \fand b \,\lor\, a \fand c \label{fop: GDJSK 50} \\
	\!\!\!\!\!\!\!\!\!\! a \fand (a \fto b) & \leq & b \label{fop: GDJSK 60} \\
	%\!\!\!\!\!\!\!\!\!\! a \fand (b \fto c) & \leq & (a \fto b) \fto c \label{fop: GDJSK 70} \\
	%\!\!\!\!\!\!\!\!\!\! a \fand (b \fequiv c) & \leq & (a \fto b) \fto c \label{fop: GDJSK 70a} \\
	%\!\!\!\!\!\!\!\!\!\! a \fand (b \fequiv c) & \leq & (a \fand c \fto b) \label{fop: GDJSK 80b} \\
	%\!\!\!\!\!\!\!\!\!\! a \fand (b \fequiv c) & \leq & b \,\fequiv\, a \fand c \label{fop: GDJSK 80} \\
	%\!\!\!\!\!\!\!\!\!\! a \fto (b \fto c) & = & b \fto (a \fto c) \label{fop: GDJSK 90} \\
	\!\!\!\!\!\!\!\!\!\! a \fto (b \fto c) & \leq & a \fand b \,\fto\, c \label{fop: GDJSK 100} \\
	\!\!\!\!\!\!\!\!\!\! a \fand (b \fto c) & \leq & b \,\fto\, a \fand c \label{fop: GDJSK 80a} \\
	%\!\!\!\!\!\!\!\!\!\! a \fto (b \fequiv c) & \leq & a \fand b \,\fto\, c \label{fop: GDJSK 110} \\
	%\!\!\!\!\!\!\!\!\!\! (a \fto b) \fand (b \fto c) & \leq & a \fto c \label{fop: GDJSK 115} \\
	%\!\!\!\!\!\!\!\!\!\! (a \fequiv b) \fand (b \fequiv c) & \leq & a \fequiv c \label{fop: GDJSK 115a} \\
	%\!\!\!\!\!\!\!\!\!\! (a \fequiv a') \land (b \fequiv b') & \leq & a \land b \,\fequiv\, a' \land b' \label{fop: GDJSK 120} \\
	%\!\!\!\!\!\!\!\!\!\! (a \fequiv a') \land (b \fequiv b') & \leq & a \lor b \,\fequiv\, a' \lor b' \label{fop: GDJSK 130} \\
	%\!\!\!\!\!\!\!\!\!\! a \fequiv b & \leq & (c \fto a) \fequiv (c \fto b) \label{fop: GDJSK 140} \\
	%\!\!\!\!\!\!\!\!\!\! a \fequiv b & \leq & (a \fto c) \fequiv (b \fto c) \label{fop: GDJSK 150} \\
	\!\!\!\!\!\!\!\!\!\! a \fequiv b & \leq & (c \fequiv a) \fequiv (c \fequiv b). \label{fop: GDJSK 150b}
	\end{eqnarray}
	If $\mL$ is complete, then the following properties hold for all $a,b \in L$ and $A,B \subseteq L$:
	\begin{eqnarray}
	a \fand \textstyle\bigvee\!B & = & \bigvee_{b \in B} (a \fand b) \label{fop: GDJSK 155} \\
	%\red{a \fto \sup B} & \red{=} & \red{\sup \{a \fto b \mid b \in B\}} \label{fop: GDJSK 200} \\
	(\textstyle\bigvee\!A) \fto b & = & \bigwedge_{a \in A} (a \fto b). \label{fop: GDJSK 210}
	\end{eqnarray}
\end{lemma}

\begin{proof}
The proofs of~\eqref{fop: GDJSK 10}--\eqref{fop: GDJSK 100} are available in~\cite{NguyenFSS2021}. The assertion~\eqref{fop: GDJSK 150b} follows directly from the assertions (16), (18) and (19) of~\cite{NguyenFSS2021}. We present below proofs for the remaining assertions \eqref{fop: GDJSK 80a}, \eqref{fop: GDJSK 155} and \eqref{fop: GDJSK 210} although such proofs may be found in the literature. 

Consider the assertion~\eqref{fop: GDJSK 80a}. 
By~\eqref{fop: GDJSK 00}, $(b \fto c) \fand b \leq c$. By~\eqref{fop: GDJSK 10}, it follows that $a \fand (b \fto c) \fand b \leq a \fand c$, which implies~\eqref{fop: GDJSK 80a} by using~\eqref{fop: GDJSK 00}. 

Consider the assertion~\eqref{fop: GDJSK 155}. By~\eqref{fop: GDJSK 10}, the RHS is clearly less than or equal to the LHS. It remains to prove the converse. By~\eqref{fop: GDJSK 00}, for every $a,b \in L$, $b \leq (a \fto a \fand b)$. By~\eqref{fop: GDJSK 20}, it follows that, for every $a \in L$ and $b \in B$, $b \leq (a \fto \bigvee_{b' \in B}(a \fand b'))$. Hence, $\bigvee\!B \leq (a \fto \bigvee_{b \in B}(a \fand b))$. By~\eqref{fop: GDJSK 00}, it follows that $a \fand \bigvee\!B \leq \bigvee_{b \in B}(a \fand b)$. 

Consider the assertion~\eqref{fop: GDJSK 210}. By~\eqref{fop: GDJSK 20}, the LHS is clearly less than or equal to the RHS. It remains to prove the converse. By~\eqref{fop: GDJSK 10} and~\eqref{fop: GDJSK 60}, for every $a \in A$ and $b \in L$, $a \fand \bigwedge_{a' \in A} (a' \fto b) \leq b$. By~\eqref{fop: GDJSK 155}, it follows that $(\bigvee\!A) \fand \bigwedge_{a \in A} (a \fto b) \leq b$. Consequently, by~\eqref{fop: GDJSK 00}, $\bigwedge_{a \in A} (a \fto b) \leq ((\bigvee\!A) \fto b)$.
\myend
\end{proof}

\begin{example}
When $L = [0,1]$, the most well-known operators $\fand$ are the t-norms named after G\"odel, {\L}ukasiewicz and product. They are specified below together with the corresponding residua~($\fto$).

\medskip

\begin{center}
%	\footnotesize
	\begin{tabular}{|c|c|c|c|}
		\hline
		& G\"odel & {\L}ukasiewicz & Product \\
		\hline
		$a \fand b$ & $\min\{a,b\}$ & $\max\{0, a+b-1\}$ & $a \cdot b$ \\
		%\hline
		%$a \fOr b$ & $\max\{a,b\}$ & $\min\{1, a+b\}$ & $a + b - a \cdot b$ \\
		\hline
		$a \fto b$ 
		& 
		\(
		\!\!\left\{
		\!\!\!\begin{array}{ll}
		1 & \!\textrm{if $a \leq b$} \\ 
		b & \!\textrm{otherwise}
		\end{array}\!\!\!
		\right.
		\)
		& $\min\{1, 1 - a + b\}$
		& 
		\(
		\!\!\left\{
		\!\!\!\begin{array}{ll}
		1 & \!\textrm{if $a \leq b$} \\ 
		b/a & \!\textrm{otherwise}
		\end{array}\!\!\!
		\right.
		\)	
		\\ \hline
	\end{tabular}
\end{center}

\medskip

All of the above specified t-norms~$\fand$ are continuous. The corresponding residuated lattices are linear and complete. 
\myend
\end{example}

From now on, let $\mL = \tuple{L, \leq, \fand, \fto, 0, 1}$ be an arbitrary complete residuated lattice. 

\subsection{Fuzzy Sets}

Given a set $X$, a function $f: X \to L$ is called a {\em fuzzy set}, as well as a {\em fuzzy subset} of $X$. If $f$ is a fuzzy subset of $X$ and $x \in X$, then $f(x)$ means the fuzzy degree in which $x$ belongs to the subset. The {\em support} of a fuzzy set $f: X \to L$ is the set $\{x \in X \mid f(x) > 0\}$. A fuzzy set is said to be {\em empty} if its support is empty. 

For $\{x_1,\ldots,x_n\} \subseteq X$ and $\{a_1,\ldots,a_n\} \subseteq L$, we write $\{x_1:a_1$, \ldots, $x_n:a_n\}$ to denote the fuzzy subset $f$ of $X$ such that $f(x_i) = a_i$ for $1 \leq i \leq n$ and $f(x) = 0$ for $x \in X \setminus \{x_1,\ldots,x_n\}$. 
Similarly, given a set $I$ of indices, $\{x_i \mid i \in I\} \subseteq X$ and $\{a_i \mid i \in I\} \subseteq L$, we write $\{x_i:a_i \mid i \in I\}$ to denote the fuzzy subset $f$ of $X$ such that $f(x_i) = a_i$ for $i \in I$ and $f(x) = 0$ for $x \in X \setminus \{x_i \mid i \in I\}$. 

Given fuzzy subsets $f$ and $g$ of $X$, we write $f \leq g$ to denote that $f(x) \leq g(x)$ for all $x \in X$. If $f \leq g$, then we say that $g$ is greater than or equal to~$f$. We write $f < g$ to denote that $f \leq g$ and $f \neq g$. 
The {\em fuzzy degree of that $f$ is a subset of $g$} is denoted by $S(f,g)$ and defined as follows
\[ S(f,g) = \bigwedge_{x \in X} (f(x) \fto g(x)). \] 
The {\em fuzzy degree of that $f$ is equal to $g$} is denoted by $E(f,g)$ and defined as follows 
\[ E(f,g) = \bigwedge_{x \in X} (f(x) \fequiv g(x)). \] 

A fuzzy subset $\varphi$ of $X \times Y$ is called a {\em fuzzy relation} between $X$ and $Y$. A fuzzy relation $\varphi$ between $X$ and itself is called a fuzzy relation on~$X$. The {\em identity relation} on $X$ is the fuzzy relation $id_X : X \times X \to L$ defined as $\{\tuple{x,x}\!:\!1 \mid x \in X\}$. 

Given $\varphi: X \times Y \to L$, the converse \mbox{$\varphi^{-1} : Y \times X \to L$} of $\varphi$ is defined by $\varphi^{-1}(y,x) = \varphi(x,y)$.

The {\em composition} of fuzzy relations \mbox{$\varphi: X \times Y \to L$} and \mbox{$\psi: Y \times Z \to L$}, denoted by $\varphi \circ \psi$, is defined to be the fuzzy relation between $X$ and $Z$ such that 
\[ (\varphi \circ \psi)(x,z) = \bigvee_{y \in Y} (\varphi(x,y) \fand \psi(y,z)). \] 

Observe that $(\varphi \circ \psi)^{-1} = \psi^{-1} \circ \varphi^{-1}$. 

Given fuzzy sets $f: X \to L$, $g: Y \to L$ and $\varphi: X \times Y \to L$, we define $(f \circ \varphi) : Y \to L$ and $(\varphi \circ g) : X \to L$ to be the fuzzy sets such that
\begin{eqnarray*}
	(f \circ \varphi)(y) & = & \bigvee_{x \in X} (f(x) \fand \varphi(x,y)) \\
	(\varphi \circ g)(x) & = & \bigvee_{y \in Y} (\varphi(x,y) \fand g(y)).
\end{eqnarray*}
Note that the composition operator $\circ$ is associative. 

Given a fuzzy set $f : X \to L$ and $\lambda \in L$, by $\lambda \land f$ we denote the fuzzy subset of $X$ defined as follows: $(\lambda \land f)(x) = \lambda \land f(x)$, for $x \in X$. 

Let $\Phi$ be a set of fuzzy relations between $X$ and $Y$. The fuzzy relation $\bigcup\Phi$ between $X$ and $Y$ is specified by: 
\[ ({\textstyle\bigcup}\Phi)(x,y) = \bigvee_{\varphi \in \Phi}\!\varphi(x,y). \] 
We write $\varphi \cup \psi$ to denote $\bigcup\{\varphi,\psi\}$. 

A fuzzy relation $\varphi$ on $X$ is 
{\em reflexive} if $id_X \leq \varphi$, 
{\em symmetric} if $\varphi = \varphi^{-1}$, 
and {\em transitive} if $\varphi \circ \varphi \leq \varphi$. 
It is a {\em fuzzy pre-order} if it is reflexive and transitive. 
It is a {\em fuzzy equivalence relation} if it is reflexive, symmetric and transitive. 

%===============================================================================

\subsection{Fuzzy Automata}

A {\em fuzzy automaton} over an alphabet $\Sigma$ (and $\mL$) is a tuple $\mA = \tuple{A, \deltaA, \sigmaA, \tauA}$, where $A$ is a non-empty set of states, $\delta^\mA : A \times \Sigma \times A \to L$ is the fuzzy transition function, $\sigmaA : A \to L$ is the fuzzy set of initial states, and $\tauA : A \to L$ is the fuzzy set of terminal states.  
For $s \in \Sigma$, we write $\deltaA_s$ to denote the fuzzy relation on $A$ such that $\deltaA_s(x,y) = \deltaA(x,s,y)$. 

A fuzzy automaton $\mA = \tuple{A, \deltaA, \sigmaA, \tauA}$ is {\em image-finite} if the support of $\sigmaA$ is finite and, for every $s \in \Sigma$ and $x \in A$, the set $\{y \in A \mid \deltaA_s(x,y) > 0\}$ is finite.

A \emph{fuzzy language} over an alphabet $\Sigma$ (and $\mL$) is a fuzzy subset of $\Sigma^*$. The \emph{fuzzy language recognized by a fuzzy automaton} $\mA = \tuple{A,\deltaA,\sigmaA,\tauA}$ is the fuzzy language $\bL(\mA)$ over $\Sigma$ such that:
\[ \bL(\mA)(s_1 s_2 \ldots s_n)=\sigmaA\circ\deltaA_{s_1}\circ\deltaA_{s_2}\circ\ldots\circ\deltaA_{s_n}\circ\tauA, \]
for $n \geq 0$ and $s_1,\ldots,s_n \in \Sigma$. 
Given $x \in A$, by $\bL(\mA,x)$ we denote $\bL(\mA_x)$, where $\mA_x$ is the fuzzy automaton that differs from $\mA$ only in that $\sigma^{\mA_x} = \{x:1\}$. 

From now on, if not stated otherwise, let $\mA = \tuple{A, \deltaA, \sigmaA, \tauA}$ and $\mAp = \tuple{A', \deltaAp, \sigmaAp, \tauAp}$ be arbitrary fuzzy automata over an alphabet~$\Sigma$. 

\subsection{Simulations and Bisimulations between Fuzzy Automata}

A~{\em simulation} (called ``forward simulation'' in~\cite{CiricIDB12}) between fuzzy automata $\mA$ and $\mAp$ is a fuzzy relation \mbox{$\varphi: A \times A' \to L$} satisfying the following conditions for all $s \in \Sigma$:
\begin{eqnarray}
\sigma^\mA & \leq & \sigmaAp \circ \varphi^{-1} \label{eq: HFHAJ 1} \\
\varphi^{-1} \circ \deltaA_s & \leq & \deltaAp_s \circ \varphi^{-1} \label{eq: HFHAJ 2} \\
\varphi^{-1} \circ \tauA & \leq & \tauAp \label{eq: HFHAJ 3}.
\end{eqnarray}

A~{\em bisimulation} (called ``forward bisimulation'' in~\cite{CiricIDB12}) between fuzzy automata $\mA$ and $\mAp$ is a fuzzy relation \mbox{$\varphi: A \times A' \to L$} such that $\varphi$ is a simulation between $\mA$ and $\mAp$ and $\varphi^{-1}$ is a simulation between $\mAp$ and $\mA$, i.e., $\varphi$ satisfies the above conditions \eqref{eq: HFHAJ 1}--\eqref{eq: HFHAJ 3} as well as the following ones for all $s \in \Sigma$:
\begin{eqnarray}
\sigmaAp & \leq & \sigmaA \circ \varphi \label{eq: HFHAJ 4} \\
\varphi \circ \deltaAp_s & \leq & \deltaA_s \circ \varphi \label{eq: HFHAJ 5} \\
\varphi \circ \tauAp & \leq & \tauA. \label{eq: HFHAJ 6} 
\end{eqnarray}

An~{\em auto-simulation} (respectively, {\em auto-bisimulation}) of a fuzzy automaton $\mA$ is a simulation (respectively, bisimulation) between $\mA$ and itself.

\subsection{Approximate Simulations and Bisimulations}

Approximate simulations and bisimulations were introduced by Stanimirovi{\' c} et al.~\cite{SMC.20} for fuzzy automata over a complete Heyting algebra. In this subsection, let $\mL = \tuple{L, \leq, \fand, \fto, 0, 1}$ be an arbitrary complete Heyting algebra and let $\lambda \in L$. 

A~{\em $\lambda$-approximate simulation} (called ``$\lambda$-approximate forward simulation'' in~\cite{SMC.20}) between fuzzy automata $\mA$ and $\mAp$ is a fuzzy relation \mbox{$\varphi: A \times A' \to L$} satisfying the following conditions for all $s \in \Sigma$:
\begin{eqnarray}
\lambda & \leq & S(\sigma^\mA, \sigmaAp \circ \varphi^{-1}) \label{eq: A-HFHAJ 1} \\
\lambda & \leq & S(\varphi^{-1} \circ \deltaA_s, \deltaAp_s \circ \varphi^{-1}) \label{eq: A-HFHAJ 2} \\
\lambda & \leq & S(\varphi^{-1} \circ \tauA, \tauAp) \label{eq: A-HFHAJ 3}.
\end{eqnarray}

A~{\em $\lambda$-approximate bisimulation} (called ``$\lambda$-approximate forward bisimulation'' in~\cite{SMC.20}) between fuzzy automata $\mA$ and $\mAp$ is a fuzzy relation \mbox{$\varphi: A \times A' \to L$} such that $\varphi$ is a $\lambda$-approximate simulation between $\mA$ and $\mAp$ and $\varphi^{-1}$ is a $\lambda$-approximate simulation between $\mAp$ and $\mA$, i.e., $\varphi$ satisfies the above conditions \eqref{eq: A-HFHAJ 1}--\eqref{eq: A-HFHAJ 3} as well as the following ones for all $s \in \Sigma$:
\begin{eqnarray}
\lambda & \leq & S(\sigmaAp, \sigmaA \circ \varphi) \label{eq: A-HFHAJ 4} \\
\lambda & \leq & S(\varphi \circ \deltaAp_s, \deltaA_s \circ \varphi) \label{eq: A-HFHAJ 5} \\
\lambda & \leq & S(\varphi \circ \tauAp, \tauA). \label{eq: A-HFHAJ 6} 
\end{eqnarray}

%===============================================================================

\section{Fuzzy Simulations between Fuzzy Automata}
\label{section: fs-4-fa}

In this section, we define and study fuzzy simulations between fuzzy automata. Apart from basic properties, we also present preservation results and the Hennessy-Milner property of such simulations. 

\begin{definition}
A~{\em fuzzy simulation} between fuzzy automata $\mA$ and $\mAp$ is a fuzzy relation \mbox{$\varphi: A \times A' \to L$} satisfying the conditions \eqref{eq: HFHAJ 2} and \eqref{eq: HFHAJ 3} for all $s \in \Sigma$. 
A~{\em fuzzy auto-simulation} of a fuzzy automaton $\mA$ is a fuzzy simulation between $\mA$ and itself. 
The {\em norm} of a fuzzy simulation $\varphi$ between fuzzy automata $\mA$ and $\mAp$, denoted by $\nZs$, is defined to be $S(\sigma^\mA, \sigmaAp \circ \varphi^{-1})$. 
\myend
\end{definition}

\begin{remark}\label{remark: HHDFS} 
The following properties are consequences of the definition. 
\begin{enumerate}
\item Every simulation between $\mA$ and $\mAp$ is also a fuzzy simulation between $\mA$ and $\mAp$, but not vice versa. 
If $\varphi$ is a simulation between $\mA$ and $\mAp$, then $\nZs = 1$. 
\item The empty fuzzy relation between $A$ and $A'$ is a fuzzy simulation between $\mA$ and $\mAp$. 
\item The identity relation $id_A$ is a fuzzy auto-simulation of~$\mA$. 
%\myend
\end{enumerate}
\end{remark}

\begin{figure}
\begin{center}
\begin{tikzpicture}[->,>=stealth,auto]
\node (S) {$\mA$};
\node (u) [node distance=0.8cm, below of=S] {$u:\{\sigma:0.7\}$};
\node (bu) [node distance=2cm, below of=u] {};		
\node (v) [node distance=1.1cm, left of=bu] {$v: \{\tau:0.6\}$};
\node (w) [node distance=1.1cm, right of=bu] {$w: \{\tau:0.7\}$};
\draw (u) to node[left]{0.5} (v);
\draw (u) to node[right]{0.8} (w);
\node (Sp) [node distance=6.7cm, right of=S] {$\mAp$};
\node (up) [node distance=0.8cm, below of=Sp] {$u':\{\sigma:0.6\}$};
\node (bup) [node distance=2cm, below of=up] {};		
\node (vp) [node distance=1.1cm, left of=bup] {$v': \{\tau:0.6\}$};
\node (wp) [node distance=1.1cm, right of=bup] {$w': \{\tau:0.7\}$};
\draw (up) to node[left]{0.8} (vp);
\draw (up) to node[right]{0.7} (wp);
\end{tikzpicture}
\end{center}
\caption{An illustration for Examples~\ref{example: HGDSJ}, \ref{example: HFUWH}, \ref{example: JHDLS}, \ref{example: HFBBA}, \ref{example: KSNAO} and \ref{example: JHHSK}.\label{fig: HFKWS}}
\end{figure}

\begin{example}\label{example: HGDSJ}
Let $L = [0,1]$ and let $\fand$ be the G\"odel t-norm. Let $\mA$ and $\mAp$ be the finite fuzzy automata over $\Sigma = \{s\}$ that are specified below and illustrated in Fig.~\ref{fig: HFKWS}.
\begin{itemize}
\item $A = \{u,v,w\}$, $\sigmaA = \{u\!:\!0.7\}$, $\tauA = \{v\!:\!0.6$, $w\!:\!0.7\}$, 
$\deltaA_s = \{\tuple{u,v}\!:\!0.5$, $\tuple{u,w}\!:\!0.8\}$;
\item $A' = \{u',v',w'\}$, $\sigmaAp = \{u'\!:\!0.6\}$, $\tauAp = \{v'\!:\!0.6$, $w'\!:\!0.7\}$, 
$\deltaAp_s = \{\tuple{u',v'}\!:\!0.8$, $\tuple{u',w'}\!:\!0.7\}$.
\end{itemize}

Let $\varphi: A \times A' \to L$ be the fuzzy relation $\{\tuple{u,u'}\!:\!0.7$, $\tuple{v,v'}\!:\!1$, $\tuple{v,w'}\!:\!1$, $\tuple{w,v'}\!:\!0.6$, $\tuple{w,w'}\!:\!1\}$. 
It is easy to check that, for every $x' \in A'$ and $y \in A$,  
\begin{eqnarray*}
(\varphi^{-1} \circ \deltaA_s)(x',y) & \leq & (\deltaAp_s \circ \varphi^{-1})(x',y) \\
(\varphi^{-1} \circ \tauA)(x') & \leq & \tauAp(x').
\end{eqnarray*}
That is, $\varphi$ is a fuzzy simulation between $\mA$ and $\mAp$. 
Since $\sigmaA = \{u\!:\!0.7\}$ and $(\sigmaAp \circ \varphi^{-1})(u) = 0.6$, we have $\nZs = (0.7 \fto 0.6) = 0.6$. 
Let $\psi$ is the greatest fuzzy simulation between $\mA$ and $\mAp$. Thus, $\psi \geq \varphi$. 
\begin{itemize}
\item Since $(\tauA(w) \fto \tauAp(v')) = (0.7 \fto 0.6) = 0.6$, by~\eqref{eq: HFHAJ 3} for~$\psi$ (instead of $\varphi$), we must have $\psi(w,v') \leq 0.6$. 
\item Since $\tauA(v) > 0$ and $\tauA(w) > 0$, while $\tauAp(u') = 0$, by~\eqref{eq: HFHAJ 3} for~$\psi$, we must have $\psi(v,u') = \psi(w,u') = 0$. 
\item Since $\deltaA_s(u,v) > 0$ and $(\deltaAp_s \circ \psi^{-1})(v',v) = 0$, by~\eqref{eq: HFHAJ 2} for~$\psi$, we must have $\psi(u,v') = 0$.   

\item Since $\deltaA_s(u,w) > 0$ and $(\deltaAp_s \circ \psi^{-1})(w',w) = 0$, by~\eqref{eq: HFHAJ 2} for~$\psi$, we must have $\psi(u,w') = 0$.   

\item Since $\deltaA_s(u,w) > 0.7$ and $(\deltaAp_s \circ \psi^{-1})(u',w) = 0.7$, by~\eqref{eq: HFHAJ 2} for~$\psi$, we must have $\psi(u,u') \leq 0.7$.   
\end{itemize}
Therefore, $\psi \leq \varphi$. This implies that $\varphi$ is the greatest fuzzy simulation between $\mA$ and $\mAp$. 
\myend
\end{example}

\begin{example}\label{example: HFUWH}
Let $L = [0,1]$ and let $\fand$ be the product t-norm. Let $\Sigma$, $\mA$ and $\mAp$ be as in Example~\ref{example: HGDSJ}. 
In a similar way as done in Example~\ref{example: HGDSJ}, it can be checked that the fuzzy relation $\varphi = \{\tuple{u,u'}\!:\!7/8$, $\tuple{v,v'}\!:\!1$, $\tuple{v,w'}\!:\!1$, $\tuple{w,v'}\!:\!6/7$, $\tuple{w,w'}\!:\!1\}$ is the greatest fuzzy simulation between $\mA$ and~$\mAp$, with $\nZs = 0.75$.
\myend
\end{example}

\subsection{Basic Properties}

As basic properties of fuzzy simulations between fuzzy automata, we have the following theorem as well as results about the relationship between the notion of fuzzy simulation and the notions of simulation~\cite{CiricIDB12} and approximate simulation~\cite{SMC.20}.

\begin{theorem}\label{theorem: HFHSA}
Let $\mA$, $\mAp$ and $\mAdp$ be fuzzy automata. 
\begin{enumerate}
\item If $\varphi$ and $\psi$ are fuzzy simulations between $\mA$ and $\mAp$ and $\varphi \leq \psi$, then $\normS{\varphi}{\mA}{\mAp} \leq \normS{\psi}{\mA}{\mAp}$.
\item If $\varphi$ is a fuzzy simulation between $\mA$ and $\mAp$ and $\psi$ is a fuzzy simulation between $\mAp$ and $\mAdp$, then $\varphi \circ \psi$ is a fuzzy simulation between $\mA$ and $\mAdp$ and 
\begin{equation}\label{eq: DJHSD}
\normS{\varphi}{\mA}{\mAp} \fand \normS{\psi}{\mAp}{\mAdp} \leq \normS{\varphi \circ \psi}{\mA}{\mAdp}.
\end{equation}
\item If $\Phi$ is a set of fuzzy simulations between $\mA$ and $\mAp$, then $\bigcup\Phi$ is also a fuzzy simulation between $\mA$ and $\mAp$. 
\item The greatest fuzzy simulation between $\mA$ and~$\mAp$ exists. 
\item The greatest fuzzy auto-simulation of $\mA$ is a fuzzy pre-order and its norm is equal to~1. 
\end{enumerate}
\end{theorem}

\begin{proof}
The first assertion follows from~\eqref{fop: GDJSK 10} and~\eqref{fop: GDJSK 20}.

Consider the second assertion. Let $\varphi$ be a fuzzy simulation between $\mA$ and $\mAp$ and $\psi$ be a fuzzy simulation between $\mAp$ and $\mAdp$. Thus, for every $s \in \Sigma$, 
\[
\begin{array}{rclcrcl}
\varphi^{-1} \circ \deltaA_s & \leq & \deltaAp_s \circ \varphi^{-1} & & 
\psi^{-1} \circ \deltaAp_s & \leq & \deltaAdp_s \circ \psi^{-1} \\
\varphi^{-1} \circ \tauA & \leq & \tauAp & &
\psi^{-1} \circ \tauAp & \leq & \tauAdp.
\end{array}
\]
Since $\circ$ is associative and $\fand$ is monotonic as shown by \eqref{fop: GDJSK 10}, for every $s \in \Sigma$, 
\begin{eqnarray*}
&	(\varphi \circ \psi)^{-1} \circ \deltaA_s 
= \psi^{-1} \circ \varphi^{-1} \circ \deltaA_s
\leq \psi^{-1} \circ \deltaAp_s \circ \varphi^{-1}\ 
\leq \deltaAdp_s \circ \psi^{-1} \circ \varphi^{-1} 
= \deltaAdp_s \circ (\varphi \circ \psi)^{-1}, \\
& (\varphi \circ \psi)^{-1} \circ \tauA  
= \psi^{-1} \circ \varphi^{-1} \circ \tauA 
\leq \psi^{-1} \circ \tauAp  
\leq \tauAdp.
\end{eqnarray*}
Therefore, $\varphi \circ \psi$ is a fuzzy simulation between $\mA$ and $\mAdp$. 
To prove~\eqref{eq: DJHSD}, by~\eqref{fop: GDJSK 10}, it suffices to show that, for every $x \in A$, 
\begin{eqnarray*}
& (\sigmaA(x) \fto (\sigmaAp \circ \varphi^{-1})(x)) \fand \normS{\psi}{\mAp}{\mAdp} \leq (\sigmaA(x) \fto (\sigmaAdp \circ \psi^{-1} \circ \varphi^{-1})(x)).
\end{eqnarray*}
By \eqref{fop: GDJSK 80a} and~\eqref{fop: GDJSK 20}, it suffices to show that, for every $x \in A$, 
\[ (\sigmaAp \circ \varphi^{-1})(x) \fand \normS{\psi}{\mAp}{\mAdp} \leq (\sigmaAdp \circ \psi^{-1} \circ \varphi^{-1})(x). \]
By \eqref{fop: GDJSK 155} and~\eqref{fop: GDJSK 10}, it suffices to show that, for every $x' \in A'$, 
\[ \sigmaAp(x') \fand \normS{\psi}{\mAp}{\mAdp} \leq (\sigmaAdp \circ \psi^{-1})(x'). \]
This holds due to the definition of $\normS{\psi}{\mAp}{\mAdp}$, \eqref{fop: GDJSK 10} and~\eqref{fop: GDJSK 60}.

Consider the third assertion. Let $\Phi$ be a set of fuzzy simulations between $\mA$ and $\mAp$. 
By using~\eqref{fop: GDJSK 155} and~\eqref{eq: HFHAJ 2}, we have that, for every $s \in \Sigma$, 
\begin{eqnarray*}
(\textstyle\bigcup\Phi)^{-1} \circ \deltaA_s 
& = & \textstyle\bigcup\{\varphi^{-1} \mid \varphi \in \Phi\} \circ \deltaA_s \\
& = & \textstyle\bigcup\{\varphi^{-1} \circ \deltaA_s \mid \varphi \in \Phi\} \\
& \leq & \textstyle\bigcup\{\deltaAp_s \circ \varphi^{-1} \mid \varphi \in \Phi\} \\
& = & \deltaAp_s \circ \textstyle\bigcup\{\varphi^{-1} \mid \varphi \in \Phi\} \\
& = & \deltaAp_s \circ (\textstyle\bigcup\Phi)^{-1}. 
\end{eqnarray*}
Similarly, by using~\eqref{fop: GDJSK 155} and~\eqref{eq: HFHAJ 3}, it can be shown that $(\bigcup\Phi)^{-1} \circ \tauA \leq \tauAp$. Therefore, $\bigcup\Phi$ is a fuzzy simulation between $\mA$ and $\mAp$. 

The fourth assertion follows from the third one. 
The fifth assertion follows from the first and second assertions and Remark~\ref{remark: HHDFS} (observe that $\normS{id_A}{\mA}{\mA} = 1$). 
\myend
\end{proof}

Apart from the first point of Remark~\ref{remark: HHDFS}, the following theorem also concerns the relationship between our notion of fuzzy simulation and the notion of (forward) simulation introduced by {\'C}iri{\'c} et al.\ in~\cite{CiricIDB12}. 

\begin{theorem}\label{theorem: JHFHA}
	Let $\mA$ and $\mAp$ be fuzzy automata.
	\begin{enumerate}
		\item A fuzzy simulation $\varphi$ between $\mA$ and $\mAp$ is a simulation between them iff $\nZs = 1$. 
		\item Let $\varphi$ be the greatest fuzzy simulation between $\mA$ and~$\mAp$. If $\nZs = 1$, then $\varphi$ is also the greatest simulation between $\mA$ and $\mAp$. Otherwise, there are no simulations between $\mA$ and $\mAp$. 
	\end{enumerate}
\end{theorem}

\begin{proof}
Consider the first assertion. We have $\nZs = 1$ iff, for every $w \in \Sigma^*$, 
\mbox{$(\sigma^\mA(w) \fto (\sigmaAp \circ \varphi^{-1})(w)) = 1$}, 
which is equivalent to 
\mbox{$\sigma^\mA(w) \leq (\sigmaAp \circ \varphi^{-1})(w)$} by~\eqref{fop: GDJSK 30}. 
That is, $\nZs = 1$ iff \eqref{eq: HFHAJ 1}	holds. The first assertion of the theorem follows from this and the fact that \eqref{eq: HFHAJ 2} and~\eqref{eq: HFHAJ 3} hold by assumption. 
The second assertion follows from the first one and the fact that, if $\psi$ is a simulation between $\mA$ and $\mAp$, then it is also a fuzzy simulation between $\mA$ and~$\mAp$ with $\normS{\psi}{\mA}{\mAp} = 1$. 
\myend
\end{proof}

By the above theorem and Examples~\ref{example: HGDSJ} and~\ref{example: HFUWH}, when $L = [0,1]$ and $\fand$ is the G\"odel or product t-norm, there are no simulations between the fuzzy automata $\mA$ and $\mAp$ specified in Example~\ref{example: HGDSJ} and illustrated in Fig.~\ref{fig: HFKWS}. 

\begin{corollary}\label{cor: HDKJW}
The greatest fuzzy auto-simulation of a fuzzy automaton $\mA$ is equal to the greatest auto-simulation of~$\mA$.
\end{corollary}

\begin{proof}
Let $\varphi$ be the greatest fuzzy auto-simulation of~$\mA$. By Remark~\ref{remark: HHDFS} and Theorem~\ref{theorem: HFHSA}, $id_A$ is a fuzzy auto-simulation of $\mA$ and $\normS{id_A}{\mA}{\mA} \leq \normS{\varphi}{\mA}{\mA}$. Observe that $\normS{id_A}{\mA}{\mA} = 1$. Hence, $\normS{\varphi}{\mA}{\mA} = 1$. By Theorem~\ref{theorem: JHFHA}, it follows that $\varphi$ is the greatest auto-simulation of~$\mA$. 
\myend
\end{proof}

The following theorem concerns the relationship between our notion of fuzzy simulation and the notion of approximate (forward) simulation introduced by Stanimirovi{\'c} et al.\ in~\cite{SMC.20} for the case where $\mL$ is a complete Heyting algebra. 

\begin{theorem}\label{theorem: GHFSA}
Suppose $\mL$ is a complete Heyting algebra. Let $\varphi$ be the greatest fuzzy simulation between fuzzy automata $\mA$ and $\mAp$. 
\begin{enumerate}
\item For every $\lambda \in L$, there exists a $\lambda$-approximate simulation between $\mA$ and $\mAp$ iff $\lambda \leq \nZs$. 
\item For $\lambda = \nZs$ and $\psi$ being the greatest $\lambda$-approximate simulation between $\mA$ and $\mAp$, we have $\varphi \leq \psi$. 
\end{enumerate}
\end{theorem}

\begin{proof}
	Consider the first assertion. For the ``if'' direction, assume that $\lambda \leq \nZs$, which means \eqref{eq: A-HFHAJ 1} holds. The conditions~\eqref{eq: A-HFHAJ 2} and~\eqref{eq: A-HFHAJ 3} follow from~\eqref{eq: HFHAJ 2} and~\eqref{eq: HFHAJ 3}, respectively. Therefore, $\varphi$ is a $\lambda$-approximate simulation between $\mA$ and $\mAp$. For the ``only if'' direction, assume that $\psi$ is a $\lambda$-approximate simulation between $\mA$ and $\mAp$. 
	Thus, for every $s \in \Sigma$, 
	\begin{eqnarray*}
		\lambda & \leq & S(\sigma^\mA, \sigmaAp \circ \psi^{-1}) \\
		\lambda & \leq & S(\psi^{-1} \circ \deltaA_s, \deltaAp_s \circ \psi^{-1}) \\
		\lambda & \leq & S(\psi^{-1} \circ \tauA, \tauAp).
	\end{eqnarray*}
	Let $\xi = \lambda \land \psi$. Since $\fand = \land$, from the above inequalities it follows that 
	\begin{eqnarray*}
		\lambda & \leq & S(\sigma^\mA, \sigmaAp \circ \xi^{-1}) \\
		\xi^{-1} \circ \deltaA_s & \leq & \deltaAp_s \circ \xi^{-1} \\
		\xi^{-1} \circ \tauA & \leq & \tauAp.
	\end{eqnarray*}
	That is, $\xi$ is a fuzzy simulation between $\mA$ and $\mAp$ with \mbox{$\lambda \leq \normS{\xi}{\mA}{\mAp}$}. Thus, $\xi \leq \varphi$ and, by Theorem~\ref{theorem: HFHSA}, it follows that $\lambda \leq \normS{\xi}{\mA}{\mAp} \leq \nZs$. 
	
	Consider the second assertion. Let $\xi: A \times A' \to L$ be the fuzzy relation defined as follows:  for $\tuple{x,x'} \in A \times A'$, $\xi(x,x') = 1$ if $\varphi(x,x') \geq \lambda$, and $\xi(x,x') = \varphi(x,x')$ otherwise. Thus, $\varphi \leq \xi$ and 
	\begin{eqnarray*}
	\lambda \land (\xi^{-1} \circ \deltaA_s) & = & \lambda \land (\varphi^{-1} \circ \deltaA_s) \\
	\lambda \land (\xi^{-1} \circ \tauA) & = & \lambda \land (\varphi^{-1} \circ \tauA).
	\end{eqnarray*}
	Since $\fand = \land$ and $\varphi$ satisfies~\eqref{eq: HFHAJ 2} and~\eqref{eq: HFHAJ 3}, it follows that 
	\begin{eqnarray*}
		\lambda & \leq & S(\xi^{-1} \circ \deltaA_s, \deltaAp_s \circ \xi^{-1}) \\
		\lambda & \leq & S(\xi^{-1} \circ \tauA, \tauAp).
	\end{eqnarray*}
	Since $S(\sigma^\mA, \sigmaAp \circ \varphi^{-1}) = \nZs = \lambda$, we also have 
	\begin{eqnarray*}
		\lambda & \leq & S(\sigma^\mA, \sigmaAp \circ \xi^{-1}). 
	\end{eqnarray*}
	Therefore, $\xi$ is a $\lambda$-approximate simulation between $\mA$ and $\mAp$. Hence, $\xi \leq \psi$, which implies $\varphi \leq \psi$ since $\varphi \leq \xi$. 
\myend
\end{proof}

\begin{example}\label{example: JHDLS}
Let $L = [0,1]$ and let $\fand$ be the G\"odel t-norm. Thus, $\mL$ is a complete Heyting algebra. Let $\mA$ and $\mAp$ be the fuzzy automata specified in Example~\ref{example: HGDSJ} and illustrated in Fig.~\ref{fig: HFKWS}. Recall that $\varphi = \{\tuple{u,u'}\!:\!0.7$, $\tuple{v,v'}\!:\!1$, $\tuple{v,w'}\!:\!1$, $\tuple{w,v'}\!:\!0.6$, $\tuple{w,w'}\!:\!1\}$ is the greatest fuzzy simulation between $\mA$ and $\mAp$ and $\nZs = 0.6$. By the proof of Theorem~\ref{theorem: GHFSA}, $\xi = \{\tuple{u,u'}\!:\!1$, $\tuple{v,v'}\!:\!1$, $\tuple{v,w'}\!:\!1$, $\tuple{w,v'}\!:\!1$, $\tuple{w,w'}\!:\!1\}$ is a $0.6$-approximate simulation between $\mA$ and $\mAp$. It can be checked that $\xi$ is the greatest $0.6$-approximate simulation between $\mA$ and $\mAp$. Note that $\varphi < \xi$.
\myend
\end{example}

\subsection{Preservation of the Recognized Language}

In this subsection, we present results stating that the fuzzy language recognized by a fuzzy automaton is fuzzily preserved by fuzzy simulations. 

Recall that, for $x \in A$, $\bL(\mA,x)$ denotes $\bL(\mA_x)$, where $\mA_x$ is the fuzzy automaton that differs from $\mA$ only in that $\sigma^{\mA_x} = \{x:1\}$. The following lemma states that, if $\tuple{x,x'} \in A \times A'$ and $\varphi$ is a fuzzy simulation between $\mA$ and $\mAp$, then the fuzzy degree in which the fuzzy language recognized by $\mA_x$ is a subset of the fuzzy language recognized by $\mAp_{x'}$ is greater than or equal to $\varphi(x,x')$.   

\begin{lemma}\label{lemma: FJHSS}
If $\varphi$ is a fuzzy simulation between fuzzy automata $\mA$ and $\mAp$, then for every $\tuple{x,x'} \in A \times A'$:
\[
	\varphi(x,x') \leq S(\bL(\mA,x), \bL(\mAp,x')).
\]
\end{lemma}

\begin{proof}
Let $\varphi$ be a fuzzy simulation between $\mA$ and $\mAp$. 
It suffices to prove that, for every $w \in \Sigma^*$, 
\[ \varphi(x,x') \leq (\bL(\mA,x)(w) \fto \bL(\mAp,x')(w)), \]
or equivalently,
\[ \varphi(x,x') \fand \bL(\mA,x)(w) \leq \bL(\mAp,x')(w). \]
Let $w = s_1\ldots s_n$, with $n \geq 0$. 
The above inequality is equivalent to
\[ \varphi(x,x') \fand (\deltaA_{s_1}\circ\ldots\circ\deltaA_{s_n}\circ\tauA)(x) \leq (\deltaAp_{s_1}\circ\ldots\circ\deltaA_{s_n}\circ\tauAp)(x'). \]
It is sufficient to prove that 
\[ \varphi^{-1}\circ\deltaA_{s_1}\circ\ldots\circ\deltaA_{s_n}\circ\tauA \leq \deltaAp_{s_1}\circ\ldots\circ\deltaA_{s_n}\circ\tauAp. \]
Since $\circ$ is associative, by~\eqref{eq: HFHAJ 2}, \eqref{eq: HFHAJ 3} and~\eqref{fop: GDJSK 10}, we have 
\begin{eqnarray*}
& & \varphi^{-1}\circ\deltaA_{s_1}\circ\ldots\circ\deltaA_{s_n}\circ\tauA \\
& \leq & \deltaAp_{s_1}\circ\varphi^{-1}\circ\deltaA_{s_2}\circ\ldots\circ\deltaA_{s_n}\circ\tauA \\
& \leq & \deltaAp_{s_1}\circ\deltaAp_{s_2}\circ\varphi^{-1}\circ\deltaA_{s_3}\circ\ldots\circ\deltaA_{s_n}\circ\tauA \\
& \leq & \ldots \\
& \leq & \deltaAp_{s_1}\circ\deltaAp_{s_2}\circ\ldots\circ\deltaAp_{s_n}\circ\varphi^{-1}\circ\tauA \\
& \leq & \deltaAp_{s_1}\circ\ldots\circ\deltaAp_{s_n}\circ\tauAp.
\end{eqnarray*}
This completes the proof.
\myend
\end{proof}

The following theorem states that, if $\varphi$ is a fuzzy simulation between $\mA$ and $\mAp$, then the fuzzy degree in which the fuzzy language recognized by $\mA$ is a subset of the fuzzy language recognized by $\mAp$ is greater than or equal to the norm of~$\varphi$.   

\begin{theorem}\label{theorem: HDFUI}
If $\varphi$ is a fuzzy simulation between fuzzy automata $\mA$ and $\mAp$, then 
\[
	\nZs \leq S(\bL(\mA), \bL(\mAp)).
\]
\end{theorem}

\begin{proof}
We need to prove that 
\[
	S(\sigma^\mA, \sigmaAp \circ \varphi^{-1}) \leq S(\bL(\mA), \bL(\mAp)).
\]
That is, we need to prove that, for every $w \in \Sigma^*$, 
\[
	S(\sigma^\mA, \sigmaAp \circ \varphi^{-1}) \leq (\bL(\mA)(w) \fto \bL(\mAp)(w)),
\]
which is equivalent to 
\[
	\bigwedge_{x \in A}\!\Phi(x) \leq \Big(\!\bigvee_{x \in A}\! (\sigma^\mA(x) \fand \bL(\mA,x)(w)) \fto \bL(\mAp)(w)\Big).
\]
where $\Phi(x)$ denotes $\sigma^\mA(x) \fto (\sigmaAp \circ \varphi^{-1})(x)$. 
By~\eqref{fop: GDJSK 210}, it suffices to show that, for every $w \in \Sigma^*$, 
\[
	\bigwedge_{x \in A}\!\Phi(x) \leq \bigwedge_{x \in A} (\sigma^\mA(x) \fand \bL(\mA,x)(w) \fto \bL(\mAp)(w)).
\]
Thus, it suffices to show that, for every $w \in \Sigma^*$ and $x \in A$, 
\[
	\Phi(x) \leq (\sigma^\mA(x) \fand \bL(\mA,x)(w) \fto \bL(\mAp)(w)).
\]
By~\eqref{fop: GDJSK 100}, it suffices to show that, for every $w \in \Sigma^*$ and $x \in A$, 
\[
	\Phi(x) \leq (\sigma^\mA(x) \fto (\bL(\mA,x)(w) \fto \bL(\mAp)(w))).
\]
By~\eqref{fop: GDJSK 20}, it suffices to show that, for every $w \in \Sigma^*$ and $x \in A$, 
\[
	(\sigmaAp \circ \varphi^{-1})(x) \leq (\bL(\mA,x)(w) \fto \bL(\mAp)(w)), 
\]
which is equivalent to 
\[ \bigvee_{x' \in A'}\!(\sigmaAp(x') \fand \varphi(x,x'))\ \leq\ (\bL(\mA,x)(w) \fto \bigvee_{x' \in A'}\!(\sigmaAp(x') \fand \bL(\mAp,x')(w))). \]
By~\eqref{fop: GDJSK 20}, it suffices to show that, for every $w \in \Sigma^*$, $x \in A$ and $x' \in A'$, 
\[
	\sigmaAp(x') \fand \varphi(x,x') \leq (\bL(\mA,x)(w) \fto \sigmaAp(x') \fand \bL(\mAp,x')(w)), 
\]
which is equivalent to 
\[
	\sigmaAp(x') \fand \varphi(x,x') \fand \bL(\mA,x)(w) \leq \sigmaAp(x') \fand \bL(\mAp,x')(w).
\]
By~\eqref{fop: GDJSK 10}, it suffices to show that, for every $w \in \Sigma^*$, $x \in A$ and $x' \in A'$, 
\[
	\varphi(x,x') \fand \bL(\mA,x)(w) \leq \bL(\mAp,x')(w), 
\]
which is equivalent to 
\[
	\varphi(x,x') \leq (\bL(\mA,x)(w) \fto \bL(\mAp,x')(w)).
\]
This inequality follows from Lemma~\ref{lemma: FJHSS}. 
\myend
\end{proof}

\subsection{The Hennessy-Milner Property}

In this subsection, we present the Hennessy-Milner property of fuzzy simulations between fuzzy automata. It is a logical characterization of the greatest fuzzy simulation between two fuzzy automata under some assumptions.  

Let $\mFs$ be the smallest set of formulas over~$\Sigma$ and~$\mL$ such that:
\begin{itemize}
\item $\tau \in \mFs$; 
\item if $s \in \Sigma$ and $w \in \mFs$, then $(s \circ w) \in \mFs$;
\item if $a \in L$ and $w \in \mFs$, then $(a \to w) \in \mFs$;
\item if $w_1, w_2 \in \mFs$, then $(w_1 \land w_2) \in \mFs$.
\end{itemize}

The meaning of such formulas is as follows. 
Given a fuzzy automaton $\mA$ (over $\Sigma$ and $\mL$), a state $x \in A$ and a formula $w \in \mFs$, the fuzzy degree in which $x$ has the property $w$ is denoted by $w^\mA(x)$. It is a value from $L$ with the following intuition:
\begin{itemize}
\item $\tau^\mA(x)$ is the degree in which $x$ is a terminal state;
\item $(s \circ w)^\mA(x)$ is the degree in which executing the action $s$ at the state $x$ may lead to a state with the property $w$;  
\item $(a \to w)^\mA(x)$ is the degree in which $w^\mA(x) \geq a$;
\item $(w_1 \land w_2)^\mA(x)$ is the degree in which $x$ has both the properties $w_1$ and $w_2$. 
\end{itemize} 
Formally, the value $w^\mA(x)$ for $w \in \mFs \setminus \{\tau\}$ and $x \in A$ is defined inductively as follows: 
\begin{eqnarray*}
(s \circ w)^\mA(x) & = & (\deltaA_s \circ w^\mA)(x) \\
(a \to w)^\mA(x) & = & a \fto w^\mA(x) \\
(w_1 \land w_2)^\mA(x) & = & w_1^\mA(x) \land w_2^\mA(x).
\end{eqnarray*}
Thus, for $w \in \mFs$, $w^\mA$ is a fuzzy subset of~$A$.

\begin{example}\label{example: HFBBA}
Let $L = [0,1]$ and let $\fand$ be the G\"odel t-norm. Let $\mA$ and $\mAp$ be the fuzzy automata specified in Example~\ref{example: HGDSJ} and illustrated in Fig.~\ref{fig: HFKWS}. We have, for example, 
\begin{eqnarray*}
(s \circ (0.7 \to \tau))^\mA(u) & = & 0.8 \\
(s \circ (0.7 \to \tau))^\mAp(u') & = & 0.7 \\
(0.7 \to (s \circ \tau))^\mA(u) & = & 1 \\
(0.7 \to (s \circ \tau))^\mAp(u') & = & 1.
\end{eqnarray*}

\vspace{-3ex}

\myend
\end{example}

The following lemma is a generalization of Lemma~\ref{lemma: FJHSS}, as it implies that, if $\varphi$ is a fuzzy simulation between fuzzy automata $\mA$ and $\mAp$, then for every $\tuple{x,x'} \in A \times A'$:
\begin{equation}\label{eq: HDJHS}
\varphi(x,x') \leq \bigwedge_{w \in \mFs}\! (w^\mA(x) \fto w^\mAp(x')).
\end{equation}
This inequality states that the formulas of $\mFs$ are fuzzily preserved by fuzzy simulations. 

\begin{lemma}\label{lemma: FJHSSx}
If $\varphi$ is a fuzzy simulation between fuzzy automata $\mA$ and $\mAp$, then for every $w \in \mFs$: 
\[
	\varphi^{-1} \circ w^\mA \leq w^\mAp.
\]
\end{lemma}

\begin{proof}
Let $\varphi$ be a fuzzy simulation between $\mA$ and $\mAp$. 
We prove the lemma by induction on the structure of~$w$. 
\begin{itemize}
\item Case $w = \tau$: The assertion follows from~\eqref{eq: HFHAJ 3}. 

\item Case $w = (s \circ u)$: By~\eqref{eq: HFHAJ 2}, \eqref{fop: GDJSK 10} and the induction assumption, we have 
\[ \varphi^{-1} \circ w^\mA = \varphi^{-1} \circ \deltaA_s \circ u^\mA \leq \deltaAp_s \circ \varphi^{-1} \circ u^\mA \leq \]
\[ \leq \deltaAp_s \circ u^\mAp = w^\mAp. \]

\item Case $w = (a \to u)$: By~\eqref{fop: GDJSK 80a}, \eqref{fop: GDJSK 20} and the induction assumption, we have  
\[ \varphi^{-1} \circ w^\mA = \varphi^{-1} \circ (a \fto u^\mA)  \leq (a \fto \varphi^{-1} \circ u^\mA) \leq (a \fto u^\mAp) = w^\mAp. \]

\item Case $w = (w_1 \land w_2)$: By~\eqref{fop: GDJSK 10} and the induction assumption, we have 
\[ \varphi^{-1} \circ w^\mA = \varphi^{-1} \circ (w_1^\mA \land w_2^\mA) \leq (\varphi^{-1} \circ w_1^\mA) \land (\varphi^{-1} \circ w_2^\mA) \leq w_1^\mAp \land w_2^\mAp = w^\mAp. \]
\end{itemize}
%This completes the proof.

\vspace{-2ex}

\end{proof}

The following theorem is about the Hennessy-Milner property of fuzzy simulations between fuzzy automata. 

\begin{theorem}\label{theorem: HGDJA}
Suppose that $\mL$ is linear and $\fand$ is continuous. 
Let $\mA$ and $\mAp$ be fuzzy automata, where $\mAp$ is image-finite. 
Let $\varphi: A \times A' \to L$ be the fuzzy relation defined as follows:
\[
	\varphi(x,x') = \bigwedge_{w \in \mFs}\! (w^\mA(x) \fto w^\mAp(x')).
\]
%for $\tuple{x,x'} \in A \times A'$. 
Then, $\varphi$ is the greatest fuzzy simulation between $\mA$ and $\mAp$. 
\end{theorem}

\begin{proof}
By the consequence~\eqref{eq: HDJHS} of Lemma~\ref{lemma: FJHSSx}, it suffices to prove that $\varphi$ is a fuzzy simulation between $\mA$ and $\mAp$. 
By definition, for every $\tuple{x,x'} \in A \times A'$, 
\[
	\varphi(x,x') \leq (\tau^\mA(x) \fto \tau^\mAp(x')), 
\]
which implies 
\(
	\varphi(x,x') \fand \tauA(x) \leq \tauAp(x'). 
\)
Therefore, \eqref{eq: HFHAJ 3} holds. To prove~\eqref{eq: HFHAJ 2}, it suffices to show that, for every $\tuple{x',y} \in A' \times A$ and $x \in A$, there exists $y' \in A'$ such that  
\[ \varphi(x,x') \fand \deltaA_s(x,y) \leq \deltaAp_s(x',y') \fand \varphi(y,y'). \]
For a contradiction, suppose that there exist $\tuple{x',y} \in A' \times A$ and $x \in A$ such that, for every $y' \in A'$,   
\[ \varphi(x,x') \fand \deltaA_s(x,y) > \deltaAp_s(x',y') \fand \varphi(y,y'). \] 
Since $\fand$ is continuous, it follows that, for every $y' \in A'$, there exists $w_{y'} \in \mFs$ such that  
\[ \varphi(x,x') \fand \deltaA_s(x,y) > \deltaAp_s(x',y') \fand (w_{y'}^\mA(y) \fto w_{y'}^\mAp(y')). \] 
Let $y'_1,\ldots,y'_n$ be all elements of $A'$ such that $\deltaAp_s(x',y') > 0$ (we use here the assumption that $\mAp$ is image-finite). For $1 \leq i \leq n$, let $u_{y'_i} = (w_{y'_i}^\mA(y) \to w_{y'_i})$. We have that, for every $1 \leq i \leq n$, $u_{y'_i}^\mA(y) = 1$ (by~\eqref{fop: GDJSK 30}) and 
\[ \varphi(x,x') \fand \deltaA_s(x,y) > \deltaAp_s(x',y'_i) \fand u_{y'_i}^\mAp(y'_i). \] 
Since $\mL$ is linear, it follows that 
\begin{equation}\label{eq: HDJAA}
\varphi(x,x') \fand \deltaA_s(x,y) > \bigvee_{1 \leq i \leq n}\!(\deltaAp_s(x',y'_i) \fand u_{y'_i}^\mAp(y'_i)).
\end{equation}
Let $w = s \circ (u_{y'_1} \land \ldots \land u_{y'_n})$. Thus, by~\eqref{fop: GDJSK 10} and~\eqref{fop: GDJSK 40}, 
\begin{eqnarray*}
w^\mA(x) & \geq & \deltaA_s(x,y) \\
w^\mAp(x') & \leq & \bigvee_{1 \leq i \leq n}\!(\deltaAp_s(x',y'_i) \fand u_{y'_i}^\mAp(y'_i)).
\end{eqnarray*}
By~\eqref{eq: HDJAA} and~\eqref{fop: GDJSK 10}, it follows that 
\( \varphi(x,x') \fand w^\mA(x) > w^\mAp(x'), \)
which is equivalent to 
\( \varphi(x,x') > (w^\mA(x) \fto w^\mAp(x')). \)
This contradicts the definition of~$\varphi$. 
\myend
\end{proof}

%===============================================================================

\section{Fuzzy Bisimulations between Fuzzy Automata}
\label{section: f-s-bs}

In this section, we define and study fuzzy bisimulations between fuzzy automata. Apart from basic properties, we also present invariance results and the Hennessy-Milner property of such bisimulations. 

\begin{definition}
A~{\em fuzzy bisimulation} between fuzzy automata $\mA$ and $\mAp$ is a fuzzy relation \mbox{$\varphi: A \times A' \to L$} such that $\varphi$ is a fuzzy simulation between $\mA$ and $\mAp$ and $\varphi^{-1}$ is a fuzzy simulation between $\mAp$ and $\mA$. 
A~{\em fuzzy auto-bisimulation} of a fuzzy automaton $\mA$ is a fuzzy bisimulation between $\mA$ and itself. 
The {\em norm} of a fuzzy bisimulation $\varphi$ between fuzzy automata $\mA$ and $\mAp$, denoted by $\nZbs$, is defined as follows
\begin{eqnarray*}
	\nZbs & = & \nZs \land \normS{\varphi^{-1}}{\mAp}{\mA} \\
	& = & S(\sigma^\mA, \sigmaAp \circ \varphi^{-1}) \land S(\sigmaAp, \sigmaA \circ \varphi).
\end{eqnarray*}

%\vspace{-2ex}

%\myend
\end{definition}

Note that a fuzzy relation \mbox{$\varphi: A \times A' \to L$} is a fuzzy bisimulation between $\mA$ and $\mAp$ iff it satisfies the conditions \eqref{eq: HFHAJ 2}, \eqref{eq: HFHAJ 3}, \eqref{eq: HFHAJ 5} and \eqref{eq: HFHAJ 6} for all $s \in \Sigma$. 

\begin{remark}\label{remark: HJFSA}
The following properties are consequences of the definition. 
\begin{enumerate}
\item Every bisimulation between $\mA$ and $\mAp$ is also a fuzzy bisimulation between $\mA$ and $\mAp$, but not vice versa. 
If $\varphi$ is a bisimulation between $\mA$ and $\mAp$, then $\nZbs = 1$. 
\item The empty fuzzy relation between $A$ and $A'$ is a fuzzy bisimulation between $\mA$ and $\mAp$. 
\item The identity relation $id_A$ is a fuzzy auto-bisimulation of~$\mA$. 
\myend
\end{enumerate}
\end{remark}

\begin{example}\label{example: KSNAO}
Let $L = [0,1]$ and let $\mA$ and $\mAp$ be the fuzzy automata specified in Example~\ref{example: HGDSJ} and illustrated in Fig.~\ref{fig: HFKWS}. In a similar way as for Example~\ref{example: HGDSJ}, it can be checked that:
\begin{itemize}
\item if $\fand$ is the G\"odel t-norm, then the fuzzy relation $\varphi = \{\tuple{u,u'}\!:\!0.6$, $\tuple{v,v'}\!:\!1$, $\tuple{v,w'}\!:\!0.6$, $\tuple{w,v'}\!:\!0.6$, $\tuple{w,w'}\!:\!1\}$ is the greatest fuzzy bisimulation between $\mA$ and $\mAp$ and $\nZbs = 0.6$;
\item if $\fand$ is the product t-norm, then the fuzzy relation $\varphi = \{\tuple{u,u'}\!:\!6/7$, $\tuple{v,v'}\!:\!1$, $\tuple{v,w'}\!:\!6/7$, $\tuple{w,v'}\!:\!6/7$, $\tuple{w,w'}\!:\!1\}$ is the greatest fuzzy bisimulation between $\mA$ and $\mAp$ and $\nZbs = 36/49$. 
\myend
\end{itemize}
\end{example}

\subsection{Basic Properties}

As basic properties of fuzzy bisimulations between fuzzy automata, we have the following theorem as well as results about the relationship between the notion of fuzzy bisimulation and the notions of bisimulation~\cite{CiricIDB12} and approximate bisimulation~\cite{SMC.20}.

\begin{theorem}\label{theorem: HFLKA}
Let $\mA$, $\mAp$ and $\mAdp$ be fuzzy automata. 
\begin{enumerate}
\item If $\varphi$ and $\psi$ are fuzzy bisimulations between $\mA$ and $\mAp$ and $\varphi \leq \psi$, then $\normBS{\varphi}{\mA}{\mAp} \leq \normBS{\psi}{\mA}{\mAp}$.
\item If $\varphi$ is a fuzzy bisimulation between $\mA$ and $\mAp$, then $\varphi^{-1}$ is a fuzzy bisimulation between $\mAp$ and $\mA$ and 
\[
	\normBS{\varphi^{-1}}{\mAp}{\mA} = \normBS{\varphi}{\mA}{\mAp}.
\] 
\item If $\varphi$ is a fuzzy bisimulation between $\mA$ and $\mAp$ and $\psi$ is a fuzzy bisimulation between $\mAp$ and $\mAdp$, then $\varphi \circ \psi$ is a fuzzy bisimulation between $\mA$ and $\mAdp$ and 
\begin{equation}\label{eq: DJHSD2}
\normBS{\varphi}{\mA}{\mAp} \fand \normBS{\psi}{\mAp}{\mAdp} \leq \normBS{\varphi \circ \psi}{\mA}{\mAdp}.
\end{equation} 
\item If $\Phi$ is a set of fuzzy bisimulations between $\mA$ and $\mAp$, then $\bigcup\Phi$ is also a fuzzy bisimulation between $\mA$ and~$\mAp$. 
\item The greatest fuzzy bisimulation between $\mA$ and $\mAp$ exists. 
\item The greatest fuzzy auto-bisimulation of $\mA$ is a fuzzy equivalence relation  and its norm is equal to~1. 
\end{enumerate}
\end{theorem}

%\begin{proof}
The first assertion follows from~\eqref{fop: GDJSK 10} and~\eqref{fop: GDJSK 20}. The second assertion follows directly from the definition of fuzzy bisimulations. The third and fourth assertions follow from the second and third assertions of Theorem~\ref{theorem: HFHSA}, respectively. In particular, \eqref{eq: DJHSD2} follows from \eqref{eq: DJHSD} and~\eqref{fop: GDJSK 10}. 
The fifth assertion follows from the fourth one. 
The sixth assertion follows from the first three assertions and Remark~\ref{remark: HJFSA} (observe that $\normBS{id_A}{\mA}{\mA} = 1$). 
%\myend
%\end{proof}

\begin{corollary}\label{cor: HDHGS}
Let $\mA_1$, $\mA_2$, $\mAp_1$ and $\mAp_2$ be fuzzy automata. Suppose that there exist bisimulations between $\mA_1$ and $\mA_2$ as well as between $\mAp_1$ and $\mAp_2$. Let $\varphi_1$ (respectively, $\varphi_2$) be the greatest fuzzy bisimulation between $\mA_1$ and $\mAp_1$ (respectively, $\mA_2$ and $\mAp_2$). Then, $\normBS{\varphi_1}{\mA_1}{\mAp_1} = \normBS{\varphi_2}{\mA_2}{\mAp_2}$.  
\end{corollary}
	
\newcommand{\ProofCorollaryHDHGS}{Let $\psi$ be a bisimulation between $\mA_1$ and $\mA_2$, and $\psi'$ a bisimulation between $\mAp_1$ and $\mAp_2$. By Remark~\ref{remark: HJFSA}, $\psi$ is a fuzzy bisimulation between $\mA_1$ and $\mA_2$ with $\normBS{\psi}{\mA_1}{\mA_2} = 1$, and similarly, $\psi'$ is a fuzzy bisimulation between $\mAp_1$ and $\mAp_2$ with $\normBS{\psi'}{\mAp_1}{\mAp_2} = 1$. By Theorem~\ref{theorem: HFLKA}, $\psi \circ \varphi_2 \circ (\psi')^{-1}$ is a fuzzy bisimulation between $\mA_1$ and $\mAp_1$. Therefore, also by Theorem~\ref{theorem: HFLKA}, 
\[
	\normBS{\varphi_2}{\mA_2}{\mAp_2} \leq \normBS{\psi \circ \varphi_2 \circ (\psi')^{-1}}{\mA_1}{\mAp_1} \leq \normBS{\varphi_1}{\mA_1}{\mAp_1}.
\]
Similarly, it can be shown that $\normBS{\varphi_1}{\mA_1}{\mAp_1} \leq \normBS{\varphi_2}{\mA_2}{\mAp_2}$. 
Therefore, $\normBS{\varphi_1}{\mA_1}{\mAp_1} = \normBS{\varphi_2}{\mA_2}{\mAp_2}$.
}

\begin{proof}
\ProofCorollaryHDHGS
\myend
\end{proof}

Apart from Corollary~\ref{cor: HDHGS} and the first point of Remark~\ref{remark: HJFSA}, the following theorem also concerns the relationship between our notion of fuzzy bisimulation and the notion of (forward) bisimulation introduced by {\'C}iri{\'c} et al.\ in~\cite{CiricIDB12}. It is a counterpart of Theorem~\ref{theorem: JHFHA}. 

\begin{theorem}\label{theorem: JHFHA2}
Let $\mA$ and $\mAp$ be fuzzy automata. 
\begin{enumerate}
\item A fuzzy bisimulation $\varphi$ between $\mA$ and $\mAp$ is a bisimulation between them iff $\nZbs = 1$. 
\item Let $\varphi$ be the greatest fuzzy bisimulation between $\mA$ and $\mAp$. If $\nZbs = 1$, then $\varphi$ is also the greatest bisimulation between $\mA$ and $\mAp$. Otherwise, there are no bisimulations between $\mA$ and $\mAp$. 
\end{enumerate}
\end{theorem}

\newcommand{\ProofTheoremJHFHAt}{Consider the first assertion. We have $\nZbs = 1$ iff $\nZs = 1$ and $\normS{\varphi^{-1}}{\mAp}{\mA} = 1$, i.e., iff the following conditions hold for every $w \in \Sigma^*$: 
	\begin{eqnarray*}
		(\sigma^\mA(w) \fto (\sigmaAp \circ \varphi^{-1})(w)) & = & 1 \\  
		(\sigma^\mAp(w) \fto (\sigmaA \circ \varphi)(w)) & = & 1.
	\end{eqnarray*}
	By~\eqref{fop: GDJSK 30}, these equalities are equivalent to 
	\begin{eqnarray*}
		\sigma^\mA(w) & \leq & (\sigmaAp \circ \varphi^{-1})(w) \\
		\sigma^\mAp(w) & \leq & (\sigmaA \circ \varphi)(w), 
	\end{eqnarray*}
	respectively. That is, $\nZbs = 1$ iff \eqref{eq: HFHAJ 1} and~\eqref{eq: HFHAJ 4} hold. The first assertion of the theorem follows from this and the fact that \eqref{eq: HFHAJ 2}, \eqref{eq: HFHAJ 3}, \eqref{eq: HFHAJ 5} and \eqref{eq: HFHAJ 6} hold by assumption. The second assertion follows from the first one and the fact that, if $\psi$ is a bisimulation between $\mA$ and $\mAp$, then it is also a fuzzy bisimulation between $\mA$ and~$\mAp$ with $\normBS{\psi}{\mA}{\mAp} = 1$.	
}

\begin{proof}
\ProofTheoremJHFHAt
\myend
\end{proof}

By the above theorem and Example~\ref{example: KSNAO}, when $L = [0,1]$ and $\fand$ is the G\"odel or product t-norm, there are no bisimulations between the fuzzy automata $\mA$ and $\mAp$ specified in Example~\ref{example: HGDSJ} and illustrated in Fig.~\ref{fig: HFKWS}. 

The following corollary is a counterpart of Corollary~\ref{cor: HDKJW}. 

\begin{corollary}\label{cor: HDKJW2}
The greatest fuzzy auto-bisimulation of a fuzzy automaton $\mA$ is equal to the greatest auto-bisimulation of~$\mA$.
\end{corollary}

\newcommand{\ProofCorollaryHDKJWt}{Let $\varphi$ be the greatest fuzzy auto-bisimulation of~$\mA$. By Remark~\ref{remark: HJFSA} and Theorem~\ref{theorem: HFLKA}, $id_A$ is a fuzzy auto-bisimulation of $\mA$ and $\normBS{id_A}{\mA}{\mA} \leq \normBS{\varphi}{\mA}{\mA}$. Observe that $\normBS{id_A}{\mA}{\mA} = 1$. Hence, $\normBS{\varphi}{\mA}{\mA} = 1$. By Theorem~\ref{theorem: JHFHA2}, it follows that $\varphi$ is the greatest auto-bisimulation of~$\mA$.
} 

\begin{proof}
\ProofCorollaryHDKJWt
\myend
\end{proof}

The following theorem concerns the relationship between our notion of fuzzy bisimulation and the notion of approximate (forward) bisimulation introduced by Stanimirovi{\'c} et al.\ in~\cite{SMC.20} for the case where $\mL$ is a complete Heyting algebra. It is a counterpart of Theorem~\ref{theorem: GHFSA} and is proved analogously. 

\begin{theorem}\label{theorem: GHFSAt}
Suppose $\mL$ is a complete Heyting algebra. Let $\varphi$ be the greatest fuzzy bisimulation between fuzzy automata $\mA$ and $\mAp$. 
\begin{enumerate}
\item For every $\lambda \in L$, there exists a $\lambda$-approximate bisimulation between $\mA$ and $\mAp$ iff $\lambda \leq \nZbs$. 
\item For $\lambda = \nZbs$ and $\psi$ being the greatest $\lambda$-approximate bisimulation between $\mA$ and $\mAp$, we have $\varphi \leq \psi$. 
\end{enumerate}
\end{theorem}

\begin{proof}
Consider the first assertion. For the ``if'' direction, assume that $\lambda \leq \nZbs$, which implies~\eqref{eq: A-HFHAJ 1} and~\eqref{eq: A-HFHAJ 4}. The conditions~\eqref{eq: A-HFHAJ 2}, \eqref{eq: A-HFHAJ 3}, \eqref{eq: A-HFHAJ 5} and \eqref{eq: A-HFHAJ 6} follow from~\eqref{eq: HFHAJ 2}, \eqref{eq: HFHAJ 3}, \eqref{eq: HFHAJ 5} and~\eqref{eq: HFHAJ 6}, respectively. Therefore, $\varphi$ is a $\lambda$-approximate bisimulation between $\mA$ and $\mAp$. For the ``only if'' direction, assume that $\psi$ is a $\lambda$-approximate bisimulation between $\mA$ and $\mAp$. 
Thus, for every $s \in \Sigma$, 
\comment{
\begin{eqnarray*}
\lambda & \leq & S(\sigma^\mA, \sigmaAp \circ \psi^{-1}) \\
\lambda & \leq & S(\psi^{-1} \circ \deltaA_s, \deltaAp_s \circ \psi^{-1}) \\
\lambda & \leq & S(\psi^{-1} \circ \tauA, \tauAp) \\[0.5ex]
\lambda & \leq & S(\sigmaAp, \sigmaA \circ \psi) \\
\lambda & \leq & S(\psi \circ \deltaAp_s, \deltaA_s \circ \psi) \\
\lambda & \leq & S(\psi \circ \tauAp, \tauA).  
\end{eqnarray*}
}
\[
\begin{array}{ll}
\lambda \leq S(\sigma^\mA, \sigmaAp \circ \psi^{-1}) &
\lambda \leq S(\sigmaAp, \sigmaA \circ \psi) \\[0.5ex]
\lambda \leq S(\psi^{-1} \circ \deltaA_s, \deltaAp_s \circ \psi^{-1})\qquad\qquad &
\lambda \leq S(\psi \circ \deltaAp_s, \deltaA_s \circ \psi) \\[0.5ex]
\lambda \leq S(\psi^{-1} \circ \tauA, \tauAp) &
\lambda \leq S(\psi \circ \tauAp, \tauA).  
\end{array}
\]
Let $\xi = \lambda \land \psi$. Since $\fand = \land$, from the above inequalities it follows that 
\comment{
\begin{eqnarray*}
\lambda & \leq & S(\sigma^\mA, \sigmaAp \circ \xi^{-1})\qquad\qquad \\
\xi^{-1} \circ \deltaA_s & \leq & \deltaAp_s \circ \xi^{-1} \\
\xi^{-1} \circ \tauA & \leq & \tauAp \\[0.5ex]
\lambda & \leq & S(\sigmaAp, \sigmaA \circ \xi) \\
\xi \circ \deltaAp_s & \leq & \deltaA_s \circ \xi \\
\xi \circ \tauAp & \leq & \tauA. 
\end{eqnarray*}
}
\[
\begin{array}{rclrcl}
\lambda & \!\!\!\leq\!\!\! & S(\sigma^\mA, \sigmaAp \circ \xi^{-1}) &
\lambda & \!\!\!\leq\!\!\! & S(\sigmaAp, \sigmaA \circ \xi) \\[0.5ex]
\xi^{-1} \circ \deltaA_s & \!\!\!\leq\!\!\! & \deltaAp_s \circ \xi^{-1} &
\xi \circ \deltaAp_s & \!\!\!\leq\!\!\! & \deltaA_s \circ \xi \\[0.5ex]
\xi^{-1} \circ \tauA & \!\!\!\leq\!\!\! & \tauAp &
\xi \circ \tauAp & \!\!\!\leq\!\!\! & \tauA. 
\end{array}
\]
That is, $\xi$ is a fuzzy bisimulation between $\mA$ and $\mAp$ with \mbox{$\lambda \leq \normBS{\xi}{\mA}{\mAp}$}. Thus, $\xi \leq \varphi$ and, by Theorem~\ref{theorem: HFLKA}, it follows that $\lambda \leq \normBS{\xi}{\mA}{\mAp} \leq \nZs$. 
	
Consider the second assertion. Let $\xi: A \times A' \to L$ be the fuzzy relation defined as follows: for $\tuple{x,x'} \in A \times A'$, $\xi(x,x') = 1$ if $\varphi(x,x') \geq \lambda$, and $\xi(x,x') = \varphi(x,x')$ otherwise. Thus, $\varphi \leq \xi$ and 
\begin{eqnarray*}
\lambda \land (\xi^{-1} \circ \deltaA_s) & \!\!\!=\!\!\! & \lambda \land (\varphi^{-1} \circ \deltaA_s) \\
\lambda \land (\xi^{-1} \circ \tauA) & \!\!\!=\!\!\! & \lambda \land (\varphi^{-1} \circ \tauA) \\[0.5ex]
\lambda \land (\xi \circ \deltaAp_s) & \!\!\!=\!\!\! & \lambda \land (\varphi \circ \deltaAp_s) \\
\lambda \land (\xi \circ \tauAp) & \!\!\!=\!\!\! & \lambda \land (\varphi \circ \tauAp). 
\end{eqnarray*}
Since $\fand = \land$ and $\varphi$ satisfies~\eqref{eq: HFHAJ 2}, \eqref{eq: HFHAJ 3}, \eqref{eq: HFHAJ 5} and \eqref{eq: HFHAJ 6}, it follows that 
\begin{eqnarray*}
\lambda & \leq & S(\xi^{-1} \circ \deltaA_s, \deltaAp_s \circ \xi^{-1}) \\
\lambda & \leq & S(\xi^{-1} \circ \tauA, \tauAp) \\[0.5ex]
\lambda & \leq & S(\xi \circ \deltaAp_s, \deltaA_s \circ \xi) \\
\lambda & \leq & S(\xi \circ \tauAp, \tauA). 
\end{eqnarray*}
Since $\nZbs = \lambda$, we also have 
\begin{eqnarray*}
\lambda & \leq & S(\sigma^\mA, \sigmaAp \circ \xi^{-1}) \\ 
\lambda & \leq & S(\sigmaAp, \sigmaA \circ \varphi)
\end{eqnarray*}
Therefore, $\xi$ is a $\lambda$-approximate bisimulation between $\mA$ and $\mAp$. Hence, $\xi \leq \psi$, which implies $\varphi \leq \psi$ since $\varphi \leq \xi$. 
%}
\myend
\end{proof}

\begin{example}\label{example: JHHSK}
Let $L = [0,1]$ and let $\fand$ be the G\"odel t-norm. Thus, $\mL$ is a complete Heyting algebra. Let $\mA$ and $\mAp$ be the fuzzy automata specified in Example~\ref{example: HGDSJ} and %illustrated in 
Fig.~\ref{fig: HFKWS}. As stated in Example~\ref{example: KSNAO}, $\varphi = \{\tuple{u,u'}\!:\!0.6$, $\tuple{v,v'}\!:\!1$, $\tuple{v,w'}\!:\!0.6$, $\tuple{w,v'}\!:\!0.6$, $\tuple{w,w'}\!:\!1\}$ is the greatest fuzzy bisimulation between $\mA$ and $\mAp$ and $\nZbs = 0.6$. By the proof of Theorem~\ref{theorem: GHFSAt}, $\xi = \{\tuple{u,u'}\!:\!1$, $\tuple{v,v'}\!:\!1$, $\tuple{v,w'}\!:\!1$, $\tuple{w,v'}\!:\!1$, $\tuple{w,w'}\!:\!1\}$ is a $0.6$-approximate bisimulation between $\mA$ and $\mAp$. It can be checked that $\xi$ is the greatest $0.6$-approximate bisimulation between $\mA$ and $\mAp$. Note that $\varphi < \xi$.
\myend
\end{example}

\subsection{Invariance of the Recognized Language}

In this subsection, we present results stating that the fuzzy language recognized by a fuzzy automaton is fuzzily invariant under fuzzy bisimulations. 

The following lemma states that, if $\tuple{x,x'} \in A \times A'$ and $\varphi$ is a fuzzy bisimulation between $\mA$ and $\mAp$, then the fuzzy degree in which the fuzzy language recognized by $\mA_x$ is equal to the fuzzy language recognized by $\mAp_{x'}$ is greater than or equal to $\varphi(x,x')$. It follows directly from Lemma~\ref{lemma: FJHSS} and the second assertion of Theorem~\ref{theorem: HFLKA}. 

\begin{lemma}\label{lemma: FJHSS2}
If $\varphi$ is a fuzzy bisimulation between fuzzy automata $\mA$ and $\mAp$, then for every $\tuple{x,x'} \in A \times A'$: 
\[
	\varphi(x,x') \leq E(\bL(\mA,x), \bL(\mAp,x')).
\]
\end{lemma}

The following theorem states that, if $\varphi$ is a fuzzy bisimulation between $\mA$ and $\mAp$, then the fuzzy degree in which the fuzzy language recognized by $\mA$ is equal to the fuzzy language recognized by $\mAp$ is greater than or equal to the norm of~$\varphi$. It follows directly from Theorem~\ref{theorem: HDFUI}. 

\begin{theorem}\label{theorem: HDFUI2}
If $\varphi$ is a fuzzy bisimulation between fuzzy automata $\mA$ and $\mAp$, then 
\(
	\nZbs \leq E(\bL(\mA), \bL(\mAp)).
\)
\end{theorem}

\subsection{The Hennessy-Milner Property}

In this subsection, we present the Hennessy-Milner property of fuzzy bisimulations between fuzzy automata. It is a logical characterization of the greatest fuzzy bisimulation between two fuzzy automata under some assumptions.  

We define the set $\mFbs$ of formulas in a similar way as for $\mFs$, but using expressions of the form $(a \leftrightarrow w)$ instead of $(a \to w)$. In particular, $\mFbs$ is the smallest set of formulas such that:
\begin{itemize}
	\item $\tau \in \mFbs$; 
	\item if $s \in \Sigma$ and $w \in \mFbs$, then $(s \circ w) \in \mFbs$;
	\item if $a \in L$ and $w \in \mFbs$, then $(a \leftrightarrow w) \in \mFbs$;
	\item if $w_1, w_2 \in \mFbs$, then $(w_1 \land w_2) \in \mFbs$.
\end{itemize}

The value $w^\mA(x)$ for $w \in \mFbs \setminus \{\tau\}$ and $x \in A$ is defined analogously as for the case where $w \in \mFs \setminus \{\tau\}$, except that $(a \leftrightarrow w)^\mA(x)$ is defined to be $(a \fequiv w^\mA(x))$, i.e., the fuzzy degree in which $w^\mA(x)$ is equal to~$a$.  

The following lemma is a counterpart of Lemma~\ref{lemma: FJHSSx} and a generalization of Lemma~\ref{lemma: FJHSS2}, as it implies that, if $\varphi$ is a fuzzy bisimulation between fuzzy automata $\mA$ and $\mAp$, then for every $x \in A$ and $x' \in A'$:
\begin{equation}\label{eq: HDJHS2}
\varphi(x,x') \leq \bigwedge_{w \in \mFbs}\! (w^\mA(x) \fequiv w^\mAp(x')).
\end{equation}
This inequality states that the formulas of $\mFbs$ are fuzzily invariant under fuzzy bisimulations. 

\begin{lemma}\label{lemma: FJHSS2x}
If $\varphi$ is a fuzzy bisimulation between fuzzy automata $\mA$ and $\mAp$, then for every $w \in \mFbs$: 
\begin{eqnarray}
\varphi^{-1} \circ w^\mA & \leq & w^\mAp \label{eq: HFDHA 1} \\
\varphi \circ w^\mAp & \leq & w^\mA. \label{eq: HFDHA 2} 
\end{eqnarray}
\end{lemma}

\newcommand{\ProofLemmaFJHSStx}{Let $\varphi$ be a fuzzy bisimulation between $\mA$ and $\mAp$. 
We prove the lemma by induction on the structure of~$w$. 
The cases where $w$ is of the form $\tau$, $(s \circ u)$ or $(w_1 \land w_2)$ can be dealt with analogously as done in the proof of Lemma~\ref{lemma: FJHSSx}. Consider the case $w = (a \leftrightarrow u)$. Let $\tuple{x,x'}$ be an arbitrary pair from $A \times A'$. 
By the induction assumption (i.e., \eqref{eq: HFDHA 1} and~\eqref{eq: HFDHA 2} with $w$ replaced by $u$), 
\[ \varphi(x,x') \leq (u^\mA(x) \fequiv u^\mAp(x')). \]
By~\eqref{fop: GDJSK 150b}, it follows that 
\[ \varphi(x,x') \leq ((a \fequiv u^\mA(x)) \fequiv (a \fequiv u^\mAp(x'))), \]
which means
\[ \varphi(x,x') \leq (w^\mA(x) \fequiv w^\mAp(x')). \]
As this holds for all $\tuple{x,x'} \in A \times A'$, we can derive \eqref{eq: HFDHA 1} and~\eqref{eq: HFDHA 2}.
}

\begin{proof}
\ProofLemmaFJHSStx
\myend
\end{proof}

The following theorem is about the Hennessy-Milner property of fuzzy bisimulations between fuzzy automata. 
It is a counterpart of Theorem~\ref{theorem: HGDJA} and is proved analogously. 

\begin{theorem}\label{theorem: HGDJAt}
Suppose that $\mL$ is linear and $\fand$ is continuous. Let $\mA$ and $\mAp$ be image-finite fuzzy automata. Let \mbox{$\varphi: A \times A' \to L$} be the fuzzy relation defined as follows:
\[
	\varphi(x,x') = \bigwedge_{w \in \mFbs}\! (w^\mA(x) \fequiv w^\mAp(x')).
\]
%for $\tuple{x,x'} \in A \times A'$. 
Then, $\varphi$ is the greatest fuzzy bisimulation between $\mA$ and $\mAp$. 
\end{theorem}

\begin{proof}
By the consequence~\eqref{eq: HDJHS2} of Lemma~\ref{lemma: FJHSS2x}, it suffices to prove that $\varphi$ is a fuzzy bisimulation between $\mA$ and $\mAp$. 
By definition, for every $\tuple{x,x'} \in A \times A'$, 
\[
	\varphi(x,x') \leq (\tau^\mA(x) \fequiv \tau^\mAp(x')), 
\]
which implies 
\begin{eqnarray*}
\varphi(x,x') \fand \tauA(x) & \leq & \tauAp(x') \\
\varphi(x,x') \fand \tauAp(x') & \leq & \tauA(x).
\end{eqnarray*}
Therefore, \eqref{eq: HFHAJ 3} and~\eqref{eq: HFHAJ 6} hold. 

Consider~\eqref{eq: HFHAJ 2}. To prove it, it suffices to show that, for every $\tuple{x',y} \in A' \times A$ and $x \in A$, there exists $y' \in A'$ such that  
\[ \varphi(x,x') \fand \deltaA_s(x,y) \leq \deltaAp_s(x',y') \fand \varphi(y,y'). \]
For a contradiction, suppose that there exist $\tuple{x',y} \in A' \times A$ and $x \in A$ such that, for every $y' \in A'$,   
\[ \varphi(x,x') \fand \deltaA_s(x,y) > \deltaAp_s(x',y') \fand \varphi(y,y'). \] 
Since $\fand$ is continuous, it follows that, for every $y' \in A'$, there exists $w_{y'} \in \mFbs$ such that  
\[ \varphi(x,x') \fand \deltaA_s(x,y) > \deltaAp_s(x',y') \fand (w_{y'}^\mA(y) \fequiv w_{y'}^\mAp(y')). \] 
Let $y'_1,\ldots,y'_n$ be all elements of $A'$ such that $\deltaAp_s(x',y') > 0$ (we use here the assumption that $\mAp$ is image-finite). For $1 \leq i \leq n$, let $u_{y'_i} = (w_{y'_i}^\mA(y) \leftrightarrow w_{y'_i})$. We have that, for every $1 \leq i \leq n$, $u_{y'_i}^\mA(y) = 1$ (by~\eqref{fop: GDJSK 30}) and 
\[ \varphi(x,x') \fand \deltaA_s(x,y) > \deltaAp_s(x',y'_i) \fand u_{y'_i}^\mAp(y'_i). \] 
Since $\mL$ is linear, it follows that 
\begin{equation}\label{eq: HDJAA2}
\varphi(x,x') \fand \deltaA_s(x,y) > \bigvee_{1 \leq i \leq n}\!(\deltaAp_s(x',y'_i) \fand u_{y'_i}^\mAp(y'_i)).
\end{equation}
Let $w = s \circ (u_{y'_1} \land \ldots \land u_{y'_n})$. Thus, by~\eqref{fop: GDJSK 10} and~\eqref{fop: GDJSK 40}, 
\begin{eqnarray*}
w^\mA(x) & \geq & \deltaA_s(x,y), \\
w^\mAp(x') & \leq & \bigvee_{1 \leq i \leq n}\!(\deltaAp_s(x',y'_i) \fand u_{y'_i}^\mAp(y'_i)).
\end{eqnarray*}
By~\eqref{eq: HDJAA2} and~\eqref{fop: GDJSK 10}, it follows that 
\[ \varphi(x,x') \fand w^\mA(x) > w^\mAp(x'), \]
which is equivalent to 
\[ \varphi(x,x') > (w^\mA(x) \fto w^\mAp(x')). \]
This contradicts the definition of~$\varphi$. 
	
The assertion~\eqref{eq: HFHAJ 5} can be proved analogously. 
%}
\myend
\end{proof}

%===============================================================================
\section{Related Work}
\label{sec: related work}

Given two fuzzy systems of the same kind such as fuzzy automata, fuzzy LTSs, fuzzy/weighted social networks, fuzzy Kripke models or fuzzy interpretations in a description logic, one can define a simulation or bisimulation between them either as a crisp relation or as a fuzzy relation between the sets of states (respectively, actors or individuals) of the systems. 
Bisimulations as crisp relations have been studied for fuzzy LTSs~\cite{CaoCK11,CaoSWC13,DBLP:journals/fss/WuD16,DBLP:journals/ijar/WuCHC18,DBLP:journals/fss/WuCBD18}, many-valued/fuzzy modal logics~\cite{EleftheriouKN12,Fan15,fuin/Diaconescu20}, weighted/fuzzy automata~\cite{DamljanovicCI14,DBLP:journals/fss/YangL20} and fuzzy description logics~\cite{FSS2020}. 
Simulations as crisp relations have been studied for fuzzy LTSs~\cite{DBLP:journals/ijar/PanLC15,DBLP:journals/fss/WuD16,fuzzIEEE/NguyenN21}, weighted/fuzzy automata~\cite{DamljanovicCI14,DBLP:journals/fss/YangL20} and fuzzy description logics \cite{Nguyen-TFS2019}. 
Bisimulations as fuzzy relations have been studied for weighted/fuzzy automata~\cite{DBLP:journals/tcs/Buchholz08,CiricIDB12,SMC.20}, fuzzy modal logics~\cite{Fan15,NguyenFSS2021}, weighted/fuzzy social networks~\cite{ai/FanL14,IgnjatovicCS15} and fuzzy description logics~\cite{FSS2020}. 
Simulations as fuzzy relations have been studied for fuzzy automata~\cite{CiricIDB12,SMC.20}, fuzzy LTSs~\cite{DBLP:journals/ijar/PanC0C14,DBLP:journals/ijar/PanLC15,ijar/Nguyen21} and fuzzy social networks~\cite{IgnjatovicCS15}. 

Notable works on simulations and bisimulations which are defined as fuzzy relations for fuzzy/weighted automata are \cite{CiricIDB12,SMC.20,DBLP:journals/tcs/Buchholz08}. The works~\cite{CiricIDB12,SMC.20} have been discussed in the introduction. In~\cite{DBLP:journals/tcs/Buchholz08} Buchholz introduced and studied bisimulations between weighted automata over a semiring. There is a similarity between the approaches of defining bisimulations in~\cite{CiricIDB12} and~\cite{DBLP:journals/tcs/Buchholz08}, but the settings based on residuated lattices~\cite{CiricIDB12} and semirings~\cite{DBLP:journals/tcs/Buchholz08} are substantially different. 

Logical characterizations of fuzzy bisimulations in fuzzy modal and description logics have been studied in~\cite{Fan15,FSS2020,NguyenFSS2021}. Logical characterizations of fuzzy simulations between fuzzy LTSs have been studied in~\cite{DBLP:journals/ijar/PanC0C14,DBLP:journals/ijar/PanLC15,ijar/Nguyen21}. 

Algorithms for computing the greatest simulation or bisimulation, which is defined as a fuzzy relation, between two finite fuzzy automata or two finite fuzzy interpretations in a description logic have been provided in~\cite{CiricIJD12,SMC.20,TFS2020}. The algorithms given in \cite{TFS2020} can be adapted to obtain algorithms with the complexity $O((m+n)n)$ for computing the greatest fuzzy simulation or bisimulation between two finite fuzzy automata, where $n$ is the number of states and $m$ is the number of non-zero transitions in the automata. 

\section{Conclusions}
\label{sec: conc}

We have introduced the notions of fuzzy simulation and bisimulation for fuzzy automata over a residuated lattice. Technically, our notion of fuzzy simulation (respectively, bisimulation) is obtained from the notion of (forward) simulation (respectively, bisimulation) introduced by {\'C}iri{\'c} et al.~\cite{CiricIDB12} for fuzzy automata by removing some conditions. Conceptually, however, this gives completely new notions which differ from the ones of~\cite{CiricIDB12} substantially. The most important difference is that (forward) simulations and bisimulations between two fuzzy automata act as a crisp relationship \cite[Theorem~5.3]{CiricIDB12}, while fuzzy simulations and bisimulations between two fuzzy automata act as a fuzzy relationship (Theorems~\ref{theorem: HDFUI} and~\ref{theorem: HDFUI2}). 
Our notions of fuzzy simulation and bisimulation for fuzzy automata are more general than the notions of (forward) simulation and bisimulation~\cite{CiricIDB12}, as every forward simulation (respectively, bisimulation) between two fuzzy automata is also a fuzzy simulation (respectively, bisimulation) between them, but not vice versa. 

As discussed in the introduction, our notions of fuzzy simulation and bisimulation for fuzzy automata are also more refined than the notions of approximate (forward) simulation and bisimulation introduced by Stanimirovi{\' c} et al.~\cite{SMC.20}. 

We have introduced the norms of a fuzzy simulation or bisimulation between two fuzzy automata. These notions are essential for expressing properties and characterizations of fuzzy simulations and bisimulations (see Theorems~\ref{theorem: HFHSA}, \ref{theorem: JHFHA}, \ref{theorem: GHFSA}, \ref{theorem: HDFUI}, \ref{theorem: HFLKA}, \ref{theorem: JHFHA2}, \ref{theorem: GHFSAt}, \ref{theorem: HDFUI2} and Corollary~\ref{cor: HDHGS}). 
We have proved that the fuzzy language recognized by a fuzzy automaton is fuzzily preserved by fuzzy simulations (Theorem~\ref{theorem: HDFUI}) and fuzzily invariant under fuzzy bisimulations (Theorem~\ref{theorem: HDFUI2}). We have also proved the Hennessy-Milner properties of fuzzy simulations and bisimulations (Theorems~\ref{theorem: HGDJA} and~\ref{theorem: HGDJAt}). They are a logical characterization of the greatest fuzzy simulation or bisimulation between two fuzzy automata. 

%===============================================================================

%\bibliography{BSfDL}
%\bibliographystyle{plain}

%===============================================================================

\end{document}